\documentclass[12pt,a4paper]{article}
\usepackage{amsfonts,amsmath,amssymb,epsf,color,array,colortbl,amscd,bbm,tensor}
\usepackage{rotating}
\usepackage{tikz}
\usetikzlibrary{arrows,decorations.pathmorphing,positioning,matrix,scopes}
\usepackage{hyperref}

\allowdisplaybreaks[1]
\usepackage{amsthm}
\theoremstyle{plain}
\newtheorem{thm}{Theorem}
\newtheorem{prop}[thm]{Proposition}

\newtheorem{definition}[thm]{Definition}
\newtheorem{remark}[thm]{Remark}

\numberwithin{equation}{section}

\renewcommand{\epsilon}{\varepsilon}

\renewcommand{\hom}{\mathrm{Hom}}
\renewcommand{\phi}{\varphi}
\renewcommand{\d}{\mathrm{d}}
\newcommand{\id}{\mathrm{id}}
\newcommand{\fuse}{\dot\otimes}
\newcommand{\im}{\mathop{\mathrm{im}}}

\title{The tensor structure on the
    representation category of the {\boldmath $\mathcal{W}_p$\unboldmath} triplet algebra}
\date{}
\begin{document}
\thispagestyle{empty}
\def\thefootnote{\fnsymbol{footnote}}

\begin{center}\LARGE
  \textbf{%
    The tensor structure on the
    representation category of the {\boldmath $\mathcal{W}_p$\unboldmath} triplet algebra
  }
\end{center}\vskip 2em
\begin{center}\large
  Akihiro Tsuchiya\footnote{Email: {\tt akihiro.tsuchiya@ipmu.jp}} \,%
  and
  Simon Wood\footnote{Email: {\tt simon.wood@ipmu.jp}}
\end{center}
\begin{center}
  Institute for the Physics and Mathematics of the Universe (IPMU)\\
  The University of Tokyo
\end{center}
\vskip 1em
\begin{center}
  December, 2011
\end{center}
\vskip 1em
\begin{abstract}
  We study the braided monoidal structure that the fusion product induces on the abelian
  category \(\mathcal{W}_p\)-mod, the category of representations of the triplet \(W\)-algebra
  \(\mathcal{W}_p\).
  The \(\mathcal{W}_p\)-algebras are a family of vertex operator algebras that form the simplest known examples of symmetry algebras
  of logarithmic conformal field theories. We formalise the methods for
  computing fusion products, developed by Nahm, Gaberdiel and Kausch,
  that are widely used in
  the physics literature and illustrate a systematic approach to calculating fusion products in non-semi-simple 
  representation categories. 
  We apply these methods to the braided monoidal structure of
  \(\mathcal{W}_p\)-mod, previously
  constructed by Huang, Lepowsky and Zhang, to prove
  that this braided monoidal structure is rigid. 
  The rigidity of \(\mathcal{W}_p\)-mod allows us to prove
  explicit formulae for the fusion product on the set of all simple and all
  projective \(\mathcal{W}_p\)-modules, which were first conjectured by
  Fuchs, Hwang, Semikhatov and Tipunin; and Gaberdiel and Runkel.
\end{abstract}

\setcounter{footnote}{0}
\def\thefootnote{\arabic{footnote}}

\newpage

\section{Introduction}

The theory of vertex operator algebras is an algebraic approach to describing the chiral symmetry algebras
of conformal field theories, at least when the number of irreducible representations of the symmetry algebra is finite
\cite{Borcherds:1983sq,Frenkel:1988xz,Frenkel:1992,Frenkel:1993,Zhu:1996,Frenkel:2001}.
Over the last few years a class of conformal field theories, called logarithmic
conformal field theories, has been the subject of a lot of research. Logarithmic conformal field theories appear
in the description of critical points of a number of interesting physical systems. Examples are polymers, spin chains, percolation and sand-pile models \cite{Jeng:2006tg,Pearce:2006we,Read:2007qq,Mathieu:2007pe,Ridout:2008cv,Nigro:2009si}.
Logarithmic conformal field theories generalise the conformal field theories most commonly considered, by
allowing the singularities encountered in correlation functions, when two field insertion approach each other, to be
logarithmic divergencies rather than just poles \cite{Gurarie:1993xq}. Two necessary consequences of the logarithmic divergencies are that
\(L_0\) the generator of scale transformations is no longer diagonalisable and that the representation theory of the
symmetry algebra is non-semi-simple. The non-semi-simplicity in particular has made it quite challenging to find
rigorous mathematical classifications of the representations of symmetry algebras associated to logarithmic conformal
field theories.

Arguably the two best understood families of vertex operator algebras associated to logarithmic conformal field theories
are the \(\mathcal{W}_p\)- and the \(\mathcal{W}_{p_+,p_-}\)-series, where \(p\geq2\) and \(p_{\pm}\geq 2\), with \(p_{+},p_-\)
coprime respectively. The \(\mathcal{W}_p\)-series is by now quite well understood 
\cite{Gaberdiel:1996np,Fuchs:2003yu,Huang:2003za,Carqueville:2005nu,Feigin:2005zx,Feigin:2006iv,Gaberdiel:2007jv,Adamovic:2007er,Bushlanov:2009cv,Ridout:2010qx}.
It was shown in \cite{Abe:2007kb} that the \(\mathcal{W}_2\)
triplet algebra satisfies Zhu's \(c_2\)-cofiniteness condition. The same
result for general \(p\geq2\) was shown in \cite{Carqueville:2005nu,Adamovic:2007er}. The representation
category \(\mathcal{W}_p\)-mod was completely classified in \cite{Nagatomo:2009xp} as a \(\mathbb{C}\)-linear abelian
category. The representation theory of \(\mathcal{W}_{p_+,p_-}\)-series is not so well understood yet.
There are well supported conjectures for lists of all irreducible and all projective representations, but there is still a
lot of work to be done
\cite{Feigin:2006iv,Feigin:2006xa,Rasmussen:2008xi,Eberle:2006zn,Gaberdiel:2010rg,Pearce:2010pa}.
For all odd \(p_-\geq 3\) is was shown in \cite{Adamovic:2009xs,Adamovic:2010zk,Adamovic:2011xq} that
the triplet algebra \(\mathcal{W}_{2,p_-}\) satisfies Zhu's \(c_2\)-cofiniteness
condition and additionally the centre of Zhu's algebra was determined.

It was shown by Huang, Lepowsky and Zhang in the series of papers 
\cite{Huang:2003za,Huang:2010pm,Huang:2010pn,Huang:2010pp,Huang:2010pq,Huang:2010pr,Huang:2010ps,Huang:2011fn,Huang:2011fp,Huang:2009xq,Huang:2009mj}
that the fusion product of
\(\mathcal{W}_p\)-representations induces a braided monoidal structure on
\(\mathcal{W}_p\)-mod, the representation category of the
\(\mathcal{W}_p\) triplet algebra.
The purpose of this paper is to analyse the fusion product of
\(\mathcal{W}_p\)-representations
by making heavy use of all that is known of
\(\mathcal{W}_p\)-mod as an abelian category.
The main results can be summarised as follows:
\begin{itemize}
\item Section \ref{sec:explicitdefoffusion} in which a systematic
  description is given on how to define and compute fusion products in a
  non-semi-simple setting.
\item Theorem \ref{sec:wprigid} which states that the braided monoidal
  structure of \(\mathcal{W}_p\)-mod induced by fusion is rigid.
\item Theorem \ref{sec:fusionrules} which gives explicit
  formulae for the fusion products of all simple and all projective
  \(\mathcal{W}_p\)-modules.
\end{itemize}
The formulae in theorem \ref{sec:fusionrules}
were first conjectured in \cite{Fuchs:2003yu,Gaberdiel:2007jv}.

As a final comment we would like to note that the \(\mathcal{W}_p\)-series is closely related to quantum groups at roots of
unity \cite{Feigin:2005zx,Feigin:2006iv}. Indeed it was shown in \cite{Nagatomo:2009xp} that the representation categories
of \(\mathcal{W}_p\) and its corresponding quantum group are equivalent as abelian categories. It was shown in \cite{Kondo:2009} that
the standard quantum group tensor product cannot coincide with the \(\mathcal{W}_p\) fusion product, because it is not braided.

The paper is organised as follows. 
In section \ref{sec:deffusion} we
introduce our notation for vertex operator algebras, give a short definition
of monoidal categories and explain how to define and compute the fusion product in the representation category \(V\)-mod
of an arbitrary \(c_2\)-cofinite vertex operator algebra \(V\). 
In section \ref{sec:defwpmod} we introduce the \(\mathcal{W}_p\) triplet algebra and its representation category \(\mathcal{W}_p\)-mod.
Sections \ref{sec:deffusion} and \ref{sec:defwpmod} are introductory and serve to familiarise the
reader with fusion products, the \(\mathcal{W}_p\)-algebra, representations of the \(\mathcal{W}_p\)-algebra and to introduce our notation.

Section \ref{sec:calcX1andX2} contains a detailed analysis of two simple \(\mathcal{W}_p\)-modules which we call \(X_1^-\) and \(X_2^+\). We calculate their
fusion products with all simple \(\mathcal{W}_p\)-modules and prove that they are rigid objects in \(\mathcal{W}_p\)-mod. This section relies heavily on
the notions discussed in sections \ref{sec:propofmoncat} and \ref{sec:explicitdefoffusion}.
In section \ref{sec:wpisrigid} we prove this paper's two main theorems \ref{sec:wprigid} and \ref{sec:fusionrules}
by exploiting the rigidity of \(X_1^-\) and \(X_2^+\) to compute the fusion product of
\(X_1^-\) and \(X_2^+\) with all projective modules. This allows us to prove the rigidity of \(\mathcal{W}_p\)-mod and to
compute the fusion product on the set of all simple and all projective modules as well as the induced product on
the Grothendieck group \(K(\mathcal{W}_p)\).

\subsection*{Acknowledgements}

This work was supported by World Premier International Research Center Initiative (WPI Initiative), MEXT, Japan.
The first author is supported by Kakenhi-grant number 22540010, while the second author is supported by the SNSF
scholarship for prospective researchers and thanks the University of Tokyo, IPMU for its hospitality.
The second author would also like to thank Matthias Gaberdiel and Ingo Runkel,
whose conjectured fusion rules inspired this paper, for helpful discussions.

\section{The definition of fusion tensor products}
\label{sec:deffusion}

\subsection{Vertex operator algebras and current algebras}

In this section we will briefly summarise our definitions and notation for vertex operator algebras. 
For a more detailed discussion see \cite{Nagatomo:2002,Matsuo:2005}.

\begin{definition}
  A tuple \((V,\Omega,T,Y)\) -- consisting of a complex vector space \(V\), two distinguished non-trivial elements 
  \(\Omega\), \(T\in V\) and a map \(Y\) -- is called a vertex operator algebra (VOA for short), if it satisfies
  the following conditions.
  \begin{enumerate}
  \item The vector space \(V\) is non-negative integer graded
    \begin{align}
      V=\bigoplus_{n=0}^\infty V[n]\,,
    \end{align}
    such that \(V[0]=\mathbb{C}\Omega\), \(\dim V[n]<\infty\ \forall n\geq 0\)
    and \(T\in V[2]\).
  \item For \(A\in V[h_A]\), the map \(Y\) is a \(\mathbb{C}\)-linear map
    \begin{align}
      Y: V&\rightarrow \mathrm{End}_{\mathbb{C}}(V)[[z,z^{-1}]]\\\nonumber
      A&\mapsto Y(A;z)=A(z)=\sum_{n\in\mathbb{Z}} A_nz^{-n-h_A}\,,
    \end{align}
    such that
    \begin{align}
      Y(A;z)\Omega-A\ \in V[[z]]z
    \end{align}
    and
    \begin{align}
      Y(\Omega;z)=\id_V\,.
    \end{align}
  \item If we set
    \begin{align}
      Y(T;z)=\sum_{n\in \mathbb{Z}}L_n z^{-n-2}\,,
    \end{align}
    then the modes \(L_n\) satisfy the commutation relations of the Virasoro algebra with
    fixed central charge \(c=c_V\)
    \begin{align}
      [L_m,L_n]=(m-n)L_{m+n}+\frac{c_V}{12}(m^3-m)\delta_{m+n,0}\,.
    \end{align}
  \item The zero mode of the Virasoro algebra \(L_0\) acts semi-simply on \(V\) and
    \begin{align}
      V[n]=\{A\in V| L_0 A= n A\}\,.
    \end{align}
  \item For any element \(A\in V\) we have
    \begin{align}
      \frac{\d}{\d z}Y(A;z)=Y(L_{-1}A;z)\,.
    \end{align}
  \item For any elements \(A,B\in V\), \(Y(A;z)\) and \(Y(B;z)\) are local with
    respect to each other, that is, there exits an \(N>0\), such that
    \begin{align}
      (z-w)^N[Y(A;z),Y(B;w)]=0
    \end{align}
    as a formal power series in 
    \(\mathrm{End}_{\mathbb{C}}(V)[[z,z^{-1},w,w^{-1}]]\).
  \item For any elements \(A\in V[h_A]\), \(B\in V\), \(Y(A;z)\) and \(Y(B;z)\)
    satisfy the operator product expansion
    \begin{align}
      Y(A;z)Y(B;w)=\sum_{n\in \mathbb{Z}}Y(A_n B;w)(z-w)^{-n-h_A}\,.
    \end{align}
  \end{enumerate}
\end{definition}
When there is no chance of confusion we will refer to a VOA just by its graded vector space
\(V\).
\begin{remark}
  By the above definition it follows that
  for \(A\in V[h_A]\) the Virasoro generators \(L_0\) and \(L_{-1}\) satisfy
    \begin{align}
      [L_{-1},A_n]&=-(n+h_A-1)A_{n-1}\\\nonumber
      [L_0,A_n]&=-n A_n\,.
    \end{align}
\end{remark}

Next we introduce a finiteness condition due to Zhu \cite{Frenkel:1992,Zhu:1996}.
\begin{definition}
  A VOA \(V\) is said to be \(c_2\)-cofinite if
  \begin{align}
    \dim_\mathbb{C}V/c_2(V)<\infty\,,
  \end{align}
  where \(c_2(V)\) the subspace of \(V\) defined by
  \begin{align}
    c_2(V)=\mathrm{span}\{A_nB|\ A\in V[h_A],\ B\in V,\ n\leq -(h_A+1)\}\,.
  \end{align}
\end{definition}
In this paper will mainly be considering \(c_2\)-cofinite VOAs. Among many
other helpful properties \(c_2\)-cofiniteness guarantees that the \(V\) has only a finite
number of irreducible representations.

The algebra of the modes of a VOA \(V\) can be understood by
using the concepts of the current Lie algebra \(\mathfrak{g}(V)\) and the current algebra \(\mathcal{U}(V)\) of \(V\).
The representation theory of a VOA \(V\) can be defined by left \(\mathcal{U}(V)\)-modules
with some extra properties, which we will explain in sequel.

Let \(V\) be a VOA.
Consider the spaces \(V^{(1)}=\bigoplus_{h\geq0}V[h]\otimes\mathbb{C}[[\xi,\xi^{-1}]](d\xi)^{1-h}\)
and \(V^{(0)}=\bigoplus_{h\geq0}V[h]\otimes\mathbb{C}[[\xi,\xi^{-1}]](d\xi)^{-h}\) as well as the \(\mathbb{C}\)-linear map
\(\nabla:V^{(0)}\rightarrow V^{(1)}\) defined by
\begin{align}\label{eq:currentconnection}
  \nabla(A\otimes f(\xi)(\d\xi)^{-h_A})= L_{-1}A\otimes f(\xi)(\d\xi)^{-h_A}
  +A\otimes \frac{\d f(\xi)}{\d\xi}(\d\xi)^{1-h_A}.
\end{align}
\begin{definition}
  Let
  \begin{align}
    \mathfrak{g}(V)=V^{(1)}/\nabla V^{(0)}\,,
  \end{align}
  then \(\mathfrak{g}(V)\) has the structure of a Lie algebra
  given by
  \begin{align}\label{eq:currentbracket}
    &[A\otimes f(\d\xi)^{1-h_A},B\otimes g(\d\xi)^{1-h_B}]=\\
    &\qquad \sum_{m=0}^{h_A+h_B-1}\frac{1}{m!}A_{m-h_A+1}B\otimes
    \frac{d^m f}{\d\xi^m}g(\d\xi)^{m+2-h_A-h_B}.\nonumber
  \end{align}
\end{definition}
For each element \(A\in V[h]\) we denote
\begin{align}
  A_{n}&=[A\otimes\xi^{n+h-1}(\d\xi)^{1-h}]\in\mathfrak{g}(V)\,.
\end{align}
and
\begin{align}
  L_n=[T\otimes \xi^{n+1}(\d\xi)^{-1}]\in\mathfrak{g}(V)\,.
\end{align}
The \(L_n\) generate the Virasoro Lie algebra as a subalgebra in \(\mathfrak{g}(V)\) and the Lie algebra
\(\mathfrak{g}(V)\) has a \(\mathbb{Z}\)-graded Lie algebra structure by
\begin{align}
  \mathfrak{g}(V)[n]=\left\{g\in\mathfrak{g}(V)\,|\, [L_0,g]=ng\right\}\,.
\end{align}
By definition the \(A_n\) are the elements in \(\mathfrak{g}(V)[-n]\).

We now define the current algebra \(\mathcal{U}(V)\) of \(V\). We first consider the universal enveloping algebra
\(U(\mathfrak{g}(V))\) of \(\mathfrak{g}(V)\). Then \(U(\mathfrak{g}(V))\) has the structure of a \(\mathbb{Z}\)-graded algebra by
decomposition into \(L_0\) eigenspaces
\begin{align}
  U(\mathfrak{g}(V))[n]=\left\{P\in U(\mathfrak{g}(V))\,|\,[L_0,P]=nP\right\}\,.
\end{align}
Consider the degreewise completion of \(U(\mathfrak{g}(V))\) 
\begin{align}
  \overline{U(\mathfrak{g}(V))}=\bigoplus_{n\in\mathbb{Z}}\overline{U(\mathfrak{g}(V))[n]}
\end{align}
and
consider the degreewise closed two sided ideal
\begin{align}
  \overline{\mathcal{I}}=\bigoplus_{n\in\mathbb{Z}}\overline{\mathcal{I}[n]}
\end{align}
of \(\overline{U(\mathfrak{g}(V))}\) generated by the Borcherds relations which arise from the operator product expansion
\begin{align}
  Y(A;z)Y(B;w)=\sum_{n\in\mathbb{Z}}Y(A_nB;w)(z-w)^{-n-h_A}\,.
\end{align}

\begin{definition}
  The current algebra \(\mathcal{U}(V)\) of \(V\) is the topological
  \(\mathbb{Z}\)-graded algebra
  \begin{align}
    \mathcal{U}(V)=\bigoplus_n \mathcal{U}(V)[n]=\bigoplus_n
    \overline{U(\mathfrak{g}(V))[n]}/\overline{\mathcal{I}[n]}\,.
  \end{align}
\end{definition}

The following proposition is very important in this paper, because
it allows us to switch back and forth between calculations in the current
Lie algebra and the current algebra.
\begin{prop}
  The canonical \(\mathbb{Z}\)-graded Lie algebra map
  \begin{align}
    \mathfrak{g}(V)=\bigoplus_n\mathfrak{g}(V)[n]\rightarrow
    \mathcal{U}(V)=\bigoplus_n\mathcal{U}(V)[n]
  \end{align}
  has a dense image.
\end{prop}

We define filtrations of \(\mathcal{U}(V)\) 
\begin{align}\label{eq:filtration}
  \mathcal{F}_k(\mathcal{U})&=\bigoplus_{n\geq k} \mathcal{U}(V)[n]\,,\\\nonumber
  \mathcal{F}^k(\mathcal{U})&=\bigoplus_{n\leq k} \mathcal{U}(V)[n]\,,
\end{align}
satisfying
\begin{align}
  \cdots\mathcal{F}_{k-1}(\mathcal{U})\supset\mathcal{F}_k(\mathcal{U})\supset\mathcal{F}_{k+1}(\mathcal{U})\cdots\,,\\\nonumber
  \cdots\mathcal{F}^{k-1}(\mathcal{U})\subset\mathcal{F}^k(\mathcal{U})\subset\mathcal{F}^{k+1}(\mathcal{U})\cdots\,.
\end{align}

The category \(V\)-mod of representations of the
VOA \(V\) is defined using left \(\mathcal{U}(V)\)-modules.
\begin{definition}
  A representation \(M\) of the VOA \(V\), also called a \(V\)-module, is
  a left \(\mathcal{U}(V)\)-module containing a finite dimensional \(\mathcal{F}_0(\mathcal{U})\) invariant subspace
  \(M_0\) such that \(\mathcal{U}(V)\cdot M_0=M\).
  We denote the abelian category of \(V\)-modules by \(V\)-mod.
\end{definition}

For any \(V\)-module \(M\) we can define the \(\mathbb{C}\)-linear map
\begin{align}
  Y^M:V&\rightarrow \mathrm{End}_{\mathbb{C}}(M)[[z,z^{-1}]]\\\nonumber
  A&\mapsto Y^M(A;z)=\sum_{n\in\mathbb{Z}}\rho_M(A_n)z^{-n-h_A}\,,
\end{align}
where \(\rho_M\) is a representation of the current Lie algebra \(\mathfrak{g}(V)\) on \(M\)
and we have assumed that \(A\in V[h_A]\). Then we have the following formulae
\begin{align}
  Y^M(\Omega;z)&=\id_M\\\nonumber
  \frac{\d}{\d z}Y^M(A;z)&=Y^M(L_{-1}A;z)
\end{align}
and for all \(A\in V[h_A]\) and \(B\in V\), the operators \(Y^M(A;z)\) and
\(Y^M(B;w)\) are local with respect to each other and satisfy the operator
product expansion
\begin{align}
  Y^M(A;z)Y^M(B;w)=\sum_{n\in\mathbb{Z}}Y^{M}(A_nB;w)(z-w)^{-n-h_A}\,.
\end{align}

\begin{definition}
  The current algebra \(\mathcal{U}(V)\) admits the algebra
  anti-automorphism
  \begin{align}
    \sigma: \mathcal{U}(V)[n]\rightarrow \mathcal{U}(V)[-n]\,,
  \end{align}
  which for \(A\in V[h_A]\) is defined by \(\sigma(A_{-n})=
  (-1)^{h_A} (e^{L_{1}}A)_{n}\).
\end{definition}

\begin{definition}
  For any \(V\)-module \(M\) the contragredient \(V\)-module \(M^\ast\) is defined by
  \begin{align}
    M^\ast=\sum_{h\in H}\hom_\mathbb{C} (M[h],\mathbb{C})
  \end{align}
  as a vector space. While the action of the current algebra is given by
  \begin{align}
    \langle P\phi,u\rangle=\langle \phi,\sigma(P)u\rangle\,,
  \end{align}
  for \(\phi\in M^\ast, u\in M, P\in \mathcal{U}(V)\). The contragredient of the contragredient is again the
  original module \(M=(M^\ast)^\ast\).
\end{definition}

\begin{prop}\label{sec:c2props}
  Let \(V\) be a \(c_2\) cofinite VOA.
  \begin{enumerate}
  \item The number of isomorphism classes of simple \(V\)-modules is finite.
  \item For any \(V\)-module \(M\) all Jordan-H\"older series
    \begin{align}\label{eq:JordanHoelder}
      M=M_0\supset M_1\supset\cdots \supset M_n=\{0\}\,,
    \end{align}
    such that \(M_i/M_{i+1}\) are simple \(V\)-modules, are finite.
  \item Each \(V\)-module \(M\) decomposes into a direct sum of finite
    dimensional generalised \(L_0\)-eigenspaces \(M[h]\)
    \begin{align}
      M&=\bigoplus_{h\in H}M[h]\,,\\\nonumber M[h]&=\left\{u\in M,
        (L_0-h)^Nu=0,\text{ for some }N\geq 1\right\}\,,
    \end{align}
    where \(H\) is the set of all weights, a discrete subset of
    \(\mathbb{C}\) generated from the finite set \(H_0\) of highest
    weights by adding all non-negative integers.
  \item \(V\)-mod admits a contravariant endofunctor called the
    contragredient or contragredient dual.
  \end{enumerate}
\end{prop}

Proposition \ref{sec:c2props} was shown in \cite{Huang:2009mj,Miamoto:2004}.

To better analyse \(V\)-mod we define finite dimensional \(\mathbb{C}\)-algebras
\({A}_k(V),\ k=0,1,\dots\) in the following way. Consider the degreewise closed right \(\mathcal{U}(V)\)-ideal
\begin{align}
  \mathcal{I}_k=\overline{\mathcal{F}_{k+1}(\mathcal{U})\cdot \mathcal{U}(V)}\subset \mathcal{U}(V)
\end{align}
and consider the two sided \(\mathcal{F}_0(\mathcal{U})\) ideal
\begin{align}
  I_k=\mathcal{I}_k\cap \mathcal{F}_0(\mathcal{U})\,.
\end{align}
Taking the quotient of \(\mathcal{F}_0(\mathcal{U})\) by \(I_k\) we define a series of 
\(\mathbb{C}\)-algebras
\begin{align}
  A_k(V)=\frac{\mathcal{F}_0(\mathcal{U})}{I_{k}}\,.
\end{align}
For the remainder of this section we will assume that the VOA \(V\) is \(c_2\)-cofinite.
\begin{prop}
  For \(k=0,1,2,\dots\) the \(\mathbb{C}\)-algebras \(A_k(V)\) are all finite dimensional.
\end{prop}

We denote by \(A_k(V)\)-mod, the abelian category of finite dimensional left \(A_k(V)\)-modules.
We define the covariant functor
\begin{align}
  A_k: V\text{-mod}&\rightarrow A_k(V)\text{-mod}\\\nonumber
  M&\rightarrow \mathcal{A}_k(M)=\frac{M}{I_k(M)}\,.
\end{align}

\begin{prop}
  \begin{enumerate}
  \item A \(V\)-module \(M\) is the zero module if and only if \(\mathcal{A}_k(M)\) is the
    zero module.
  \item If \(M\) is a simple \(V\)-module, then \(\mathcal{A}_k(M)\) is simple and the
    set of isomorphism classes of simple \(V\)-modules are in one to one correspondence with
    the isomorphism classes of simple \({A}_k(V)\)-modules.
  \end{enumerate}
\end{prop}
The \(k=0\) case is most important to our work at hand, we will refer to \(A_0(V)\) as the
zero mode algebra of the VOA \(V\).

To define the notion of the fusion tensor product on \(V\)-mod, we prepare some additional
concepts and definitions. As a first step we define left and right completions of the current
algebra \(\mathcal{U}(V)\). For each \(k\in\mathbb{Z}\) define
\begin{align}
  \mathcal{F}^k(\mathcal{U}^L)&=\prod_{d\leq k}\mathcal{U}(V)[d]
  &\mathcal{F}_k(\mathcal{U}^R)&=\prod_{d\geq k}\mathcal{U}(V)[d]\\\nonumber
  \mathcal{U}^L&=\bigcup_k\mathcal{F}^k(\mathcal{U}^L)
  &\mathcal{U}^R&=\bigcup_k\mathcal{F}_k(\mathcal{U}^R)\,.
\end{align}
Then \(\mathcal{U}^L\) and \(\mathcal{U}^R\) are topological \(\mathbb{C}\)-algebras with
topologies defined by the filtrations \(\mathcal{F}^k(\mathcal{U}^L)\) and \(\mathcal{F}_k(\mathcal{U}^R)\)
and the canonical inclusions
\begin{align}
  \mathcal{U}(V)&\rightarrow \mathcal{U}^L\\\nonumber
  \mathcal{U}(V)&\rightarrow \mathcal{U}^R
\end{align}
have dense images.

For any object \(M\) of \(V\)-mod, we define its closure by
\begin{align}
  \overline{M}=\varprojlim_{k}M/\mathcal{F}_k(\mathcal{U})(M)\,.
\end{align}
Then there is a continuous action of \(\mathcal{U}^R\) on \(\overline{M}\). Note that \(M\)
already has the structure of a \(\mathcal{U}^L\)-module.
By the properties of projective limits \(\overline{M}\) is equipped with
a complete Hausdorff linear topology. Then the action of the current
algebra \(\mathcal{U}(V)\) is uniquely extended to an action of \(\mathcal{U}^R\) on \(\overline{M}\). The two spaces
\(M\) and \(\overline{M}\) share the same generalised \(L_0\)-eigenspaces
\(M[h]=\overline{M}[h]\)
\begin{align}
  M[h]&=\{m\in M|\ (L_0-h)^nm=0,\ \mathrm{for\ some\ } n\geq 1\}\\\nonumber
  \overline{M}[h]&=\{m\in \overline{M}|\ (L_0-h)^nm=0,\ \mathrm{for\ some\ } n\geq 1\}\,,
\end{align}
but unlike \(M\), \(\overline{M}\) also contains infinite sums of generalised \(L_0\)-eigenvectors
\begin{align}
  \overline{M}=\prod_{h\in H}M[h]\,.
\end{align}
The image of the canonical inclusion \(M\rightarrow \overline{M}\) is dense.

For any \(V\)-module \(M\) there is a canonical surjective linear map
\begin{align}
  \mathcal{A}_{k+1}(M)\rightarrow \mathcal{A}_k(M)
\end{align}
and the projective limit
\begin{align}
  \varprojlim_{k}\mathcal{A}_k(M)
\end{align}
has a unique continuous \(\mathcal{U}^R\) action. Indeed
\begin{align}
  \overline{M}=\varprojlim_{k}\mathcal{A}_k(M)
\end{align}
as a continuous \(\mathcal{U}^R\)-module.
We also have the canonical homomorphisms
\begin{align}\label{eq:completionsofg}
  \mathfrak{g}^L(V)&\rightarrow \mathcal{U}^L\\\nonumber
  \mathfrak{g}^R(V)&\rightarrow \mathcal{U}^R\,,
\end{align}
which both have dense images.

For each \(k\in \mathbb{Z}\) we define
\begin{align}
  \mathfrak{g}_k(V)=\sum_{d\geq k}\mathfrak{g}(V)[d]\,.
\end{align}
Then the canonical map
\begin{align}
  \mathfrak{g}_k(V)\rightarrow \mathcal{F}_k(\mathcal{U})\rightarrow \mathcal{F}_k(\mathcal{U}^R)
\end{align}
has dense image. So we have the \(\mathbb{C}\)-linear isomorphism
\begin{align}
  \frac{M}{\mathfrak{g}_{k+1}(V)(M)}\rightarrow \frac{M}{I_k(M)}=\mathcal{A}_k(M)\,.
\end{align}
Therefore we have
\begin{align}
  \overline{M}=\varprojlim_{k}\frac{M}{\mathfrak{g}_{k+1}(V)(M)}
\end{align}

For later use we define for each \(V\)-module \(M\) the quotient
\begin{align}
  \frac{M}{c_1(M)}=\frac{M}{\mathrm{span}\{A_{-n}m\,|\, m\in M,\ A\in V[h],\ n\geq h>0\}}\,.
\end{align}
Then \(M/c_1(M)\) is a finite dimensional complex vector space
and there exists a canonical surjective linear map
\begin{align}
  \frac{M}{c_1(M)}\rightarrow \mathcal{A}_0(M)\,.
\end{align}
Also for later use we introduce the following notation.
\begin{definition}\label{sec:subspaces}
  Let \(M\) be a \(V\)-module, then we define \(L_0\)-graded subspaces \(M^0\) and \(M^s\), such that
    \begin{align}
    M&=M^0\oplus \mathfrak{g}_1(V)\cdot M\,,\\\nonumber
    M&=M^s\oplus c_1(M)\,, 
  \end{align}
  respectively called the \emph{zero mode} and the \emph{special subspace},
  such that the canonical maps
  \begin{align}
    M^0&\rightarrow \mathcal{A}_0(M)\,,\\\nonumber
    M^s&\rightarrow \frac{M}{c_1(M)}\,,
  \end{align}
 are \(\mathbb{C}\)-linear isomorphisms.
\end{definition}
\begin{remark}
  The subspaces \(M^0\) and \(M^s\) are not uniquely defined. We will fix specific
  choices of subspaces in a later section.
\end{remark}

\subsection{General properties of braided monoidal categories}
\label{sec:propofmoncat}

In this section we introduce the concepts of monoidal categories, their rigidness and
some general properties. We mainly follow the appendices of the seminal papers due to
Kazhdan and Lusztig \cite[Appendix A]{Kazhdan:1994} as well as the standard reference for monoidal
categories \cite{Joyal:1993}. We will only be considering monoidal categories that are also
\(\mathbb{C}\)-linear and abelian and we assume that the reader is familiar with basic notions of
abelian categories such as exact sequences, projective modules, injective modules, etc.

\begin{definition}
  A monoidal category is a tuple \((\mathcal{C}, \otimes, \mathbf{1}, \alpha, \lambda,\rho)\) -- just \((\mathcal{C},\otimes,\mathbf{1})\) for short --
  where \(\mathcal{C}\) is a category, \(\otimes : \mathcal{C} \times \mathcal{C} \rightarrow \mathcal{C}\) is the tensor product bi-functor, 
  \(\mathbf{1} \in \mathcal{C}\) is the tensor unit, 
  \(\alpha_{L,M,N} : L \otimes (M \otimes N) \overset{\sim}{\rightarrow} (L\otimes M)\otimes N\) 
  is the associator, \(\lambda_M : \mathbf{1} \otimes M \rightarrow M\) is the left unit isomorphism, 
  and \(\rho_M : M \otimes \mathbf{1} \rightarrow M\) is the right unit isomorphism. These data are 
  subject to conditions, in particular \(\alpha\) satisfies the pentagon axiom and 
  \(\lambda\), \(\rho\), \(\alpha\) obey the triangle axiom. 
\end{definition}

\begin{definition}
  We say that an object \(M\) is weakly rigid if the contravariant
  functor
  \begin{align}
    F_M(-)=\hom(-\otimes M,\mathbf{1})
  \end{align}
  is representable, {\it i.e.} for all objects \(N\) there exits an
  object \(M^\vee\) called the tensor dual\footnote{Strictly speaking \(M^\vee\) is called the right dual.
    There is a similar notion of a left dual. However if the tensor
    category is braided then these notions are related. We will therefore only be considering
  right duals.} such
  that
  \begin{align}\label{eq:weaklyrigid}
    \hom(N\otimes M,1)\cong\hom(N,M^\vee)\,.
  \end{align}
\end{definition}
Therefore if \(M\) is weakly rigid there exits a morphism
\begin{align}
  e_M:M^\vee\otimes M\rightarrow \mathbf{1}
\end{align}
that is isomorphic to \(\text{id}_{M^\vee}\in\hom(M^\vee,M^\vee)\) by the equivalence
\eqref{eq:weaklyrigid}.
A monoidal category is called weakly rigid if all its object are weakly rigid.

\begin{definition}\label{sec:defrigid}
  An object \(M\) is said to be rigid of it is weakly rigid and there exits a morphism
  \begin{align}
    i_M:\mathbf{1}\rightarrow M\otimes M^\vee\,,
  \end{align}
  such that
  \begin{align}\label{eq:rigiditycond}
    \id_M&=\rho_M\circ(\id_M\otimes e_M)\circ \alpha_{M,M^\vee,M}^{-1}\circ(i_M\otimes\id_M)\circ \lambda_M^{-1}\\\nonumber
    \id_{M^\vee}&=\lambda_{M^\vee}\circ(e_M\otimes\id_{M^\vee})\circ\alpha_{M,M^\vee,M}\circ(\id_{M^\vee}\otimes i_M)\circ\rho_{M^\vee}^{-1}\,.
  \end{align}
\end{definition}

\begin{definition}
  A braiding \(b\) on a monoidal category \((\mathcal{C},\otimes,\mathbf{1})\) is a natural
  transformation between the functors \(\otimes\) and \(\otimes\circ P\), where
  \(P:\mathcal{C}\times \mathcal{C}\rightarrow \mathcal{C}\times\mathcal{C}\) is the
  permutation \((M,N)\rightarrow (N,M)\). For \(M, N\) in \(\mathcal{C}\), \(b\) defines
  a morphism  \(b_{M,N}\in\hom(M\otimes N,N\otimes M)\) satisfying:
  \begin{enumerate}
  \item  All \(b_{M,N}\) are isomorphisms.
  \item For any \(L,M,N\) in \(\mathcal{C}\) we have
    \begin{align}
      b_{L\otimes M,N}&=\alpha_{N,L,M}^{-1}\circ (b_{L,N}\circ \mathrm{id}_M)
      \circ \alpha_{L,N,M}\circ (\mathrm{id}_L\otimes b_{M,N})\otimes \alpha_{L,M,N}^{-1}\\\nonumber
      b_{L,M\otimes N}&=\alpha_{M,N,L}\circ (\mathrm{id}_M\otimes b_{L,N})
      \circ \alpha_{M,L,N}^{-1}\circ (b_{L,M}\circ \mathrm{id}_N)\otimes \alpha_{L,M,N}\,.
    \end{align}
  \item For all \(M\) in \(\mathcal{C}\)
    \begin{align}
      b_{M,\mathbf{1}}=b_{\mathbf{1},M}=\mathrm{id}_M\,.
    \end{align}
  \end{enumerate}
\end{definition}

\begin{prop}\label{sec:moncatprop}
  Let \((\mathcal{C},\otimes,\mathbf{1})\) be a monoidal category, then
  \begin{enumerate}
  \item For any rigid object \(M\) in \(\mathcal{C}\) the functors
    \begin{align}
      M\otimes -,\ -\otimes M:\mathcal{C}\rightarrow \mathcal{C}
    \end{align}
    are exact.
  \item For a rigid object \(M\) with dual \(M^\vee\) and arbitrary objects \(N,L\) we have the isomorphism
    \begin{align}
      \hom(N,M\otimes L)\cong \hom(M^\vee\otimes N,L)
    \end{align}
  \item Let \(M\) in \(\mathcal{C}\) be rigid with dual \(M^\vee\). 
    Then for any projective object \(P\) \(M^\vee\otimes P\) is projective. For any injective
    object \(I\),  and \(M\otimes I\) is injective in \(\mathcal{C}\) if \(M^\vee\) is also rigid.
  \item Assume that
    \begin{enumerate}
    \item The abelian category \(\mathcal{C}\) has enough projective
      and injective objects.
    \item All projective objects are injective and all injective
      objects are projective.
    \item All projective objects are rigid.
    \end{enumerate}
    Then if
    \begin{align}
      0\longrightarrow L\longrightarrow M\longrightarrow
      N\longrightarrow 0
    \end{align}
    is an exact sequence in \(\mathcal{C}\) such that two of \(L,M,N\)
    are rigid, then the third object is also rigid.
  \item If \(M,N\) in \(\mathcal{C}\) are rigid objects, then
    \(M\otimes N\) is also rigid and its dual \((M\otimes N)^\vee\) is
    given by
    \begin{align}
      (M\otimes N)^\vee=N^\vee\otimes M^\vee\,.
    \end{align}
  \end{enumerate}
\end{prop}

\subsection{Fusion tensor products and their properties}
\label{sec:explicitdefoffusion}

Fusion plays a central role in analysing conformal field theories and is indeed the
central theme of this paper. Fusion describes the short distance expansion of two
fields on the level of representations. The fusion product \(M\fuse N\) of two \(V\)-modules
\(M\) and \(N\) is the smallest \(V\)-module in which all the fields appearing in the short distance expansions
-- of fields transforming in \(M\) with fields transforming in \(N\) -- transform.

There is a wealth of literature on fusion tensor products in both mathematics and physics. In the case of rational conformal field theory
the representation category is semi-simple and the theory of fusion tensor products is well established \cite{Tsuchiya:1987,DiFrancesco:1997nk,Huang:2008}.
If the conformal field theory is logarithmic, the representation category is not semi-simple. Fortunately there are computational methods in
physics one can fall back on
for defining fusion tensor products without assuming semi-simplicity \cite{Nahm:1994by,Gaberdiel:1993mt}.
For VOAs arising from affine Lie algebras, there exists a mathematically rigorous definition of fusion tensor products due to
Kazhdan and Lusztig \cite{Kazhdan:1994} that does not rely on
semi-simplicity. In the series of papers
\cite{Huang:2003za,Huang:2010pm,Huang:2010pn,Huang:2010pp,Huang:2010pq,Huang:2010pr,Huang:2010ps,Huang:2011fn,Huang:2011fp,Huang:2009xq,Huang:2009mj}
by Huang, Lepowsky and Zhang a non-meromorphic operator
product expansion and braided tensor category structure were
constructed for representation categories of VOAs satisfying natural,
general hypotheses, without the assumption of semi-simplicity.
In this paper we will state our definition of fusion in the spirit of 
\cite{Kazhdan:1994,Tsuchiya:1987,DiFrancesco:1997nk,Nahm:1994by,Gaberdiel:1993mt}
in the case of \(c_2\)-cofinite VOAs without assuming semi-simplicity.

We fix a \(c_2\)-cofinite VOA \(V\) and prepare some notation.
\begin{definition}
  We define the current Lie algebra on the Riemann sphere with punctures at \(0,1,\infty\) by 
  \begin{align}
    \mathfrak{g}^{\mathbb{P}}(V)=\frac{\bigoplus_{h=0}^{\infty}V[h]\otimes\mathbb{C}[z,z^{-1},(z-1)^{-1}]\d z^{1-h}}{
    \nabla(\bigoplus_{n=0}^{\infty}V[h]\otimes\mathbb{C}[z,z^{-1},(z-1)^{-1}]\d z^{-h})}\,,
  \end{align}
  where \(\nabla\) is defined as in \eqref{eq:currentconnection}
  \begin{align}
    \nabla(A\otimes f(z)\d z^{-h_A})=L_{-1}A\otimes f(z)\d
    z^{-h_A}+A\otimes\frac{\d f(z)}{\d z}\d z^{1-h_A}\,.
  \end{align}
  Then \(\mathfrak{g}^{\mathbb{P}}(V)\) has the structure of a Lie algebra given by
  \begin{align}
    &[A\otimes f(z)\d z^{1-h_A},B\otimes g(z)\d z^{1-h_B}]=\\\nonumber
    &\qquad\sum_{m=0}^\infty\frac{1}{m!} A_{m-h_A+1}B\otimes\frac{\d^m f(z)}{\d z^m}g(z)\d z^{m+2-h_A-h_B}\,,
  \end{align}
  for \(A\in V[h_A]\) and \(B\in V[h_B]\).
\end{definition}
 
For \(f(z)\in \mathbb{C}[z,z^{-1},(z-1)^{-1}]\) we denote the Laurent expansions at 0, 1 and infinity by
\(f_0(\xi_0)\in\mathbb{C}((\xi_0)),\ f_1(\xi_1)\in\mathbb{C}((\xi_1))\) and \(f_\infty(\xi_\infty)
\in\mathbb{C}((\xi_\infty^{-1}))\) respectively. For example for
\begin{align}
  f(z)=\frac{1}{z-1}
\end{align}
the expansions and radii of convergence are given by
\begin{align}
  f_0(\xi_0)&=-\sum_{n\geq0} \xi_0^n\,,&\xi_0&=z\,,\quad 1>|z|>0\,,\\\nonumber
  f_1(\xi_1)&=\xi_1^{-1}\,,&\xi_1&=z-1\,,\quad 1>|z-1|>0\,,\\\nonumber
  f_\infty(\xi_\infty)&=\sum_{n\geq 0}\xi_\infty^{-(n+1)}\,,&\xi_\infty&=z\,,\quad |z|>1\,.
\end{align}

We define Lie algebra homomorphisms
\begin{align}
  j_a^L:\mathfrak{g}^{\mathbb{P}}(V)&\rightarrow \mathfrak{g}^L\,,\quad a=0,1\,,\\\nonumber
  j_\infty^R:\mathfrak{g}^{\mathbb{P}}(V)&\rightarrow \mathfrak{g}^R\,,
\end{align}
where \(\mathfrak{g}^L\) and \(\mathfrak{g}^R\) are the left and right completions 
of \(\mathfrak{g}(V)\) defined in \eqref{eq:completionsofg}, by
\begin{align}\label{eq:localexpansions}
  j_a^L([A\otimes f(z)\d z^{1-h_A}])&=[A\otimes f_a(\xi_a)\d\xi_a^{1-h_A}]\,,\\\nonumber
  j_\infty^R([A\otimes f(z)\d z^{1-h_A}])&=[A\otimes f_\infty(\xi_\infty)
  \d \xi_\infty^{1-h_A}]\,.
\end{align}
The Lie algebra maps \(j_a^L\) and \(j_\infty^R\) have dense images.

The analogue of \(\mathfrak{g}_k(V)\) for the current Lie algebra on the Riemann sphere is 
given by
\begin{align}
  \mathfrak{g}_k^{\mathbb{P}}(V)=\mathrm{span}\{[A\otimes f(z)\d z^{1-h_A}]\in \mathfrak{g}^{\mathbb{P}}(V)|\ 
  \mathrm{ord}_\infty (f(z))\leq h_A-1-k\}\,,
\end{align}
where \(\mathrm{ord}_\infty(f(z))\) is the order of the pole of \(f(z)\) at infinity. The image of the map
\begin{align}
  \mathfrak{g}_k^\mathbb{P}(V)\overset{j_\infty^R}{\rightarrow}\mathfrak{g}^R_k\rightarrow\mathcal{F}_k(\mathcal{U}^R)
\end{align}
is dense and therefore the canonical map
\begin{align}
  \frac{\mathfrak{g}^\mathbb{P}(V)}{\mathfrak{g}_k^\mathbb{P}(V)} \rightarrow\frac{\mathfrak{g}^R}{\mathfrak{g}^R_k}
\end{align}
is an isomorphism, hence
\begin{align}
  \varprojlim_k\frac{\mathfrak{g}^\mathbb{P}(V)}{\mathfrak{g}_k^\mathbb{P}(V)} 
  \rightarrow\varprojlim_k\frac{\mathfrak{g}^R}{\mathfrak{g}^R_k}
  =\mathfrak{g}^R\,.
\end{align}

Consider the map
\begin{align}
  j_{1,0}=j_1^L\otimes\mathbbm{1}+\mathbbm{1}\otimes j_0^L:\mathfrak{g}^{\mathbb{P}}(V)\rightarrow \mathfrak{g}^L\otimes \mathbbm{1}+\mathbbm{1}\otimes \mathfrak{g}^L\,.
\end{align}
For any two \(V\)-modules \(M,N\), the vector space \(M\otimes N\) is a left \(\mathfrak{g}^{\mathbb{P}}(V)\)-module
by \(j_{1,0}\).

\begin{prop}
  For each \(k=0,1,2,\dots\)
  \begin{enumerate}
  \item \(\dim_{\mathbb{C}} M\otimes N/\mathfrak{g}_k^{\mathbb{P}}(V)(M\otimes
    N)<\infty\)
  \item The Lie algebra
    \(\mathfrak{g}^R=\varprojlim_k\mathfrak{g}^{\mathbb{P}}(V)/\mathfrak{g}_k^{\mathbb{P}}(V)\)
    acts continuously on the projective limit
    \begin{align}
      \overline{M\fuse N}=\varprojlim_k \frac{M\otimes
        N}{\mathfrak{g}^{\mathbb{P}}_{k}(V)(M\otimes N)}
    \end{align}
    by \(j_\infty^R:\mathfrak{g}^{\mathbb{P}}(V)\rightarrow \mathfrak{g}^R\).
  \end{enumerate}
  Furthermore by \(\mathfrak{g}^R\rightarrow \mathcal{U}^R\) the right
  completion \(\mathcal{U}^R\) of the current algebra acts continuously on
  \(\overline{M\fuse N}\).
\end{prop}

For each \(h\in\mathbb{C}\), let
\begin{align}
   M\fuse N[h]&=\{m\in \overline{M\fuse N}| \ \exists n\geq 1\ \text{s.t.}\ (L_0-h)^nm=0\}\,.
\end{align}
The fusion product of \(M\) and \(N\) is given by
\begin{align}
  M\fuse N=\bigoplus_{h\in \mathbb{C}}M\fuse N[h]\,.
\end{align}

\begin{prop}
  \begin{enumerate}
  \item The space \(M\fuse N\) is a \(V\)-module.
  \item For each \(k\geq 0\) we have the \(\mathbb{C}\)-linear isomorphisms
    \begin{align}
      \mathcal{A}_k(M\fuse N)=\frac{M\fuse N}{I_k(M\fuse N)}\cong
      \frac{M\otimes N}{\mathfrak{g}^{\mathbb{P}}_{k+1}(M\otimes N)}\,.
    \end{align}
  \end{enumerate}
\end{prop}
Most notably we have the \(A_0(V)\)-module isomorphism
\begin{align}
  \mathcal{A}_{0}(M\fuse N)\cong\frac{M\otimes N}{\mathfrak{g}_1^{\mathbb{P}}(M\otimes N)}\,.
\end{align}

\begin{thm}\label{sec:tensorcat}
  \begin{enumerate}
  \item The triplet \((V\mathrm{-mod},\fuse, V)\) is a braided
    monoidal category with unit object \(V=\mathbf{1}\).
  \item For any \(V\)-module \(M\), the contravariant functor from \(V\)-mod to
    \(\mathbb{C}\)-Vec the category of complex vector spaces
    \begin{align}
      F_M:V\mathrm{-mod}&\rightarrow \mathbb{C}\mathrm{-Vec}\\\nonumber
      N&\mapsto F_M(N)=\hom_{V\mathrm{-mod}}(N\fuse M,V^\ast)\,,
    \end{align}
    where \(V^\ast\) is the contragredient of \(V\), is represented by \(M^\ast\) the contragredient of \(M\),
    {\it i.e.}
    \begin{align}
      F_M(N)\cong \hom_{V\mathrm{-mod}}(N,M^\ast)\,.
    \end{align}
  \end{enumerate}
\end{thm}

Theorem \ref{sec:tensorcat} was proved in the series of papers
\cite{Huang:2003za,Huang:2010pm,Huang:2010pn,Huang:2010pp,Huang:2010pq,Huang:2010pr,Huang:2010ps,Huang:2011fn,Huang:2011fp,Huang:2009xq,Huang:2009mj}
(originally announced in \cite{Huang:2003za}) by Huang, Lepowsky and
Zhang.\footnote{Theorem 4.13 in \cite{Huang:2009mj} states that Theorem
\ref{sec:tensorcat} holds for any \(c_2\)-cofinite VOA \(V\) satisfying \(\dim V[0]=1, \dim
V[n]=0, n<0\). Theorem \ref{sec:tensorcat} for the case \(V=\mathcal{W}_p\) was
explicitly spelled out in \cite[Section 5.2]{Huang:2009xq}.}
A new proof of Theorem \ref{sec:tensorcat} for \(c_2\)-cofinite VOAs on Riemann
surfaces of arbitrary genus is in preparation \cite{Hashimoto:2012}.

Unfortunately the definition of \(M\fuse N\) is rather difficult to work with, because 
even though the image of canonical map
\begin{align}
  M\otimes N\rightarrow \overline{M\fuse N}
\end{align}
is dense, it generally does not lie in \(M\fuse N\).

However for each \(k\geq0\) we can make use of the
isomorphism
\begin{align}
  \mathcal{A}_k(M\fuse N)\cong\frac{M\otimes N}{\mathfrak{g}^{\mathbb{P}}_{k+1}(V)(M\otimes N)}\,.
\end{align}
By analysing these quotients, we can study  \(M\fuse N\) level for level.
For any given element \(m\otimes n\in M\otimes N\), we will denote the class
that it represents in
\(\mathcal{A}_k(M\fuse N)\) by \([m\otimes n]\).

As we shall see in the sequel, for the purposes of this paper it will be
sufficient to make statements about \(\mathcal{A}_0(M\fuse N)\) and in one case \(\mathcal{A}_1(M\fuse N)\). 
As a vector space \(\mathfrak{g}^{\mathbb{P}}_{1}(V)\) is spanned by elements of the form
\begin{align}
  [v\otimes z^{-n+h-1}\d z^{1-h}], \quad v\in V[h]
\end{align}
for \(n\geq 0\) and
\begin{align}
  [v\otimes (z-1)^{-m+h-1}\d z^{1-h}], \quad v\in V[h]
\end{align}
for \(m\geq h\).
From the expansions defined above it therefore follows that in \(\mathcal{A}_0(M\dot\otimes N)\)
for \(1\leq n\leq h-1\) we have the relations
\begin{align}\label{eq:lowdiagaction}
  &j_{1,0}([v\otimes z^{-n+h-1}\d z^{1-h}])=\\\nonumber
  &\qquad\qquad\sum_{k=0}^{h-1-n}\binom{h-1-n}{k}v_{k-(h-1)}\otimes\mathbbm{1}+\mathbbm{1}\otimes v_{-n}=0\,,
\end{align}
and for \(m\geq h\)
\begin{align}\label{eq:diagaction}
  &j_{1,0}([v\otimes z^{-m+h-1}\d z^{1-h}])=\\\nonumber
  &\qquad\qquad\sum_{k\geq0}\binom{m-h+k}{m-h}(-1)^kv_{k-(h-1)}\otimes\mathbbm{1}+\mathbbm{1}\otimes v_{-m}=0\,,\\\nonumber
  &j_{1,0}([v\otimes (z-1)^{-m+h-1}\d z^{1-h}])=\\\nonumber
  &\qquad\qquad v_{-m}\otimes\mathbbm{1}+\sum_{k\geq0}\binom{m-h+k}{m-h}(-1)^{h-1-m}\mathbbm{1}\otimes v_{k-(h-1)}=0\,.
\end{align}
The action of the zero modes is given by
\begin{align}
  j_{1,0}(v_0)&=j_{1,0}([v\otimes z^{h-1}\d z^{1-h}])\\\nonumber
  &=\sum_{k=0}^{h-1}\binom{h-1}{k}v_{k-(h-1)}\otimes\mathbbm{1}+\mathbbm{1}\otimes v_0\,.
\end{align}
For the generators of the Virasoro algebra this means
\begin{align}\label{eq:vircomult}
  L_{-1}\otimes \mathbbm{1}&\simeq -\mathbbm{1}\otimes L_{-1}\\\nonumber
  L_{-n}\otimes \mathbbm{1}&\simeq \sum_{j=0}^\infty\binom{n-2+j}{n-2}(-1)^n\mathbbm{1}\otimes L_{j-1}\\\nonumber
  \mathbbm{1}\otimes L_{-n}&\simeq-\sum_{j=0}^\infty \binom{n-2+j}{n-2}(-1)^j L_{j-1}\otimes\mathbbm{1}\,,
\end{align}
for \(n\geq 2\) and
\begin{align}\label{eq:L0action}
  j_{1,0}(L_0)=L_{-1}\otimes\mathbbm{1}+L_0\otimes\mathbbm{1}+\mathbbm{1}\otimes L_0\,.
\end{align}

To aid us in computing \(\mathcal{A}_0(M\fuse N)\), we make use of the special and zero mode subspaces in
definition \ref{sec:subspaces} to state the following proposition due to Nahm \cite{Nahm:1994by}.
\begin{prop}\label{sec:Akestimate}
  Let \(M\) and \(N\) be \(V\)-modules. Then the canonical \(\mathbb{C}\)-linear map
  \begin{align}
    M^s\otimes N^0\rightarrow \mathcal{A}_0(M\fuse N)\rightarrow 0\,.
  \end{align}
  is a surjective. 
\end{prop}

\begin{proof}
  Let \(m\otimes n\in M\otimes N\), \(A\in V[h_A]\), \(B\in V[h_B]\), \(k\geq h_A\) and \(\ell>0\).
  We introduce two kinds manipulations by using the formulae in~\eqref{eq:diagaction}
  \begin{enumerate}
  \item Moving modes to the right
    \begin{align}\label{eq:moveright}
      [A_{-k}m\otimes n]=-\sum_{j\geq 0}\binom{k-h_A+j}{k-h_A}[m\otimes A_{j-(h_A-1)}n]\,.
    \end{align}
    Note that \(A_{-k}\) is replaced by modes with a mode number that is greater than \(-k\), {\it i.e.} the grading is
    lowered.
  \item Moving modes to the left
    \begin{align}\label{eq:moveleft}
      [m\otimes B_{-\ell}n]=-\sum_{j\geq 0}\binom{h_B-1-\ell}{j}[B_{j-(h_B-1)}m\otimes n]\,.
    \end{align}
    Note that \(B_{-\ell}\) is replaced by modes with a mode number that is greater than or equal to \(-\ell\), {\it i.e.} the grading is
    lowered or stays the same.
  \end{enumerate}

  Since the image of the canonical map \(\mathfrak{g}_k(V)\rightarrow \mathcal{F}_1(V)\) is dense and
  \(\mathcal{F}_1(V)(M)=I_0(M)\), any element \(n\) in \(N\) is represented by \(n=x\cdot n_0\) for
  some \(x\in\mathfrak{g}_1(V)\) and \(n_0\in {N}^0\). Then by using formula \eqref{eq:moveleft}, the class \([m\otimes n]\) can be represented
  by \(m^\prime\otimes n_0\) for some \(m^\prime\in M\). Consider \(m\in c_1(M)\) and \(n_0\in {N}^0\), we can assume that \(m\) is
  homogeneous such that \(m\in c_1(M)[h]\) for some \(h\). By definition we can assume that \(m\) has the form
  \begin{align}
    m=A_{-k}m_0
  \end{align}
  for some \(m_0\in M^s\) and \(A\in V[h_A],\ k\geq h_A\). Then by using formula \eqref{eq:moveright}
  \begin{align}
    [m\otimes n_0]=-\sum_{j\geq 0}\binom{k-h_A+j}{k-h_A}[m_0\otimes A_{j-(h_A-1)}n_0]\,.
  \end{align}
  For each summand we use again the fact that the class \([m_0\otimes A_{j-(h_A-1)}n_0]\)
  can be represented by an element \(m^\prime_0\otimes n_0^\prime\in M\otimes {N}^0\). If we decompose \(m_0^\prime\) into homogeneous
  summands, then the weights of the individual summands will all be less then the original weight \(h\)
  of \(m\). Because the weights are bounded from below, a finite number of applications of the formulae
  \eqref{eq:moveright} and \eqref{eq:moveleft} will yield a representative in \({M}^s\otimes {N}^0\) 
  for any class \([m\otimes n]\).
\end{proof}

Before we end this section on the fusion product we consider the relation between the
fusion product \(\fuse_{V}\) of a VOA \(V\) and the fusion product \(\fuse_{V^\prime}\)
of a subVOA \(V^\prime\subset V\).
\begin{prop}\label{sec:subfusion}
  Let \(V^\prime\) be a \(c_2\)-cofinite subVOA of the VOA \(V\).
  Let \(M\) and \(N\) be \(V\)-modules, then they are also \(V^\prime\)-modules and
  there exits a surjective \(V^\prime\)-module map
  \begin{align}
    M\fuse_{V^\prime}N\rightarrow M\fuse_{V}N\,.
  \end{align}
\end{prop}
\begin{proof}
  Since
  \begin{align}
    \mathfrak{g}_k(V^\prime)\subset \mathfrak{g}_k(V)
  \end{align}
  there is a canonical surjection of \(A_k(V^\prime)\)-modules
  \begin{align}\label{eq:subfusesurj}
    \mathcal{A}_k(M\fuse_{V^\prime}N)=\frac{M\otimes N}{\mathfrak{g}^\mathbb{P}_{k+1}(V^\prime)}
    \rightarrow \frac{M\otimes N}{\mathfrak{g}^\mathbb{P}_{k+1}(V)}=\mathcal{A}_k(M\fuse_{V}N)\,.
  \end{align}
  Note that for sufficiently large \(k\), a given generalised \(L_0\)-eigenspace is stable,
  {\it i.e.}
  \begin{align}
    \mathcal{A}_k(M\fuse_{V}N)[h]=\mathcal{A}_{k+1}(M\fuse_{V}N)[h]=
    (M\fuse_{V}N)[h],\ k>>0\,.
  \end{align}
  Therefore the surjection \eqref{eq:subfusesurj} implies a surjection
  \begin{align}
    M\fuse_{V^\prime}N[h]\rightarrow M\fuse_{V}N[h]
  \end{align}
  between generalised \(L_0\)-eigenspaces. This can be repeated for all values of \(h\) and 
  therefore there exists a surjective \(V^{\prime}\)-module map
  \begin{align}
    M\fuse_{V^\prime}N\rightarrow M\fuse_{V}N\,.
  \end{align}
\end{proof}

\begin{prop}
  Let \(M\) be a \(V\)-module, then the covariant functors \(M\fuse -\) and \(-\fuse M\) are
  right exact.
\end{prop}
\begin{proof}
  We prove the proposition for \(M\fuse -\). The proof for \(-\fuse M\) follows analogously.

  Let \(A,B,C\in V\)-mod satisfy the exact sequence
  \begin{align}
    0\rightarrow A\rightarrow B\rightarrow C\rightarrow 0\,.
  \end{align}
  Then the sequence
  \begin{align}
    M\fuse A\rightarrow M\fuse B\rightarrow M\fuse C\rightarrow 0\,.
  \end{align}
  is exact if the restriction to the generalised \(L_0\)-eigenspaces
  \begin{align}
    M\fuse A[h]\rightarrow M\fuse B[h]\rightarrow M\fuse C[h]\rightarrow 0
  \end{align}
  is exact.   Because for sufficiently large \(k\) the generalised \(L_0\)-eigenspaces for fixed generalised
  eigenvalue \(h\) are stable
  under taking \(\mathcal{A}_k\) quotients, we consider the sequence
  \begin{align}
    \frac{M\otimes A}{\mathfrak{g}^\mathbb{P}_{k}(V)(M\otimes A)}[h]\rightarrow 
    \frac{M\otimes B}{\mathfrak{g}^\mathbb{P}_{k}(V)(M\otimes B)}[h]\rightarrow
    \frac{M\otimes C}{\mathfrak{g}^\mathbb{P}_{k}(V)(M\otimes C)}[h]\rightarrow 0\,,
  \end{align}
  which is clearly exact.
\end{proof}

\subsection{Algebra morphisms between products of modules}
\label{sec:intertwiners}

We have defined the fusion tensor product in \(V\)-mod. Now we introduce the concept of vertex operators
in a conformal field theory on \(\mathbb{P}\) associated to \(V\)-mod, by extending the notions in \cite{Tsuchiya:1987}.

For a \(V\)-module \(M\) we have defined the topological completion
\begin{align}
  \overline{M}=\prod_{h\in H}M[h]\,.
\end{align}
For any two \(V\)-modules \(M,N\) we denote by \(\hom_{\mathbb{C}}^c(\overline{M},\overline{N})\), the space
of continuous \(\mathbb{C}\)-linear maps from \(\overline{M}\) to \(\overline{N}\). Then we have a 
\(\mathbb{C}\)-linear isomorphism
\begin{align}
  \hom_{\mathcal{U}(V)}(M,N)\cong \hom_{\mathcal{U}^R}(\overline{M},\overline{N}).
\end{align}
Now for the \(V\)-modules \(L,M,N\) consider the complex vector spaces
\begin{align}
  \hom_{\mathcal{U}(V)}(M\fuse N,L)\cong\hom_{\mathcal{U}^R}(\overline{M\fuse N},\overline{L})\,.
\end{align}
We know that 
\begin{align}
  \overline{M\fuse N}=\varprojlim_k\frac{M\otimes N}{\mathfrak{g}_{k+1}^\mathbb{P}(M\otimes N)}\,.
\end{align}
So there exists an injective \(\mathbb{C}\)-linear map
\begin{align}
  \hom_{\mathcal{U}(V)}(M\fuse N,L)\rightarrow \hom_{\mathbb{C}}^c(\overline{M\fuse N},\overline{L})
  \rightarrow \hom_{\mathbb{C}}(M\otimes N,\overline{L})\,.
\end{align}
The following proposition characterises the image of this map.
\begin{prop}
  For \(\psi\in\hom_{\mathbb{C}}(M\otimes N,\overline{L})\) the necessary and sufficient condition for \(\psi\) to lie in the 
  image of the map from \(\hom_{\mathcal{U}(V)}(M\fuse N,L)\) is that for all \(f(z)\in\mathbb{C}[z,z^{-1},(z-1)^{-1}]\),
  \(A\in V[h_A]\), \(m\in M\) and \(n\in N\)
  \begin{align}
    j_\infty^R([A\otimes f(z) \d z^{1-h_A}])(\psi(m)n)
    =&
    \ \psi(j_1^L([A\otimes f(z) \d z^{1-h_A}])m)n\\\nonumber
    &+\psi(m)j_0^L([A\otimes f(z) \d z^{1-h_A}])n\,,
  \end{align}
  where we denote
  \begin{align}
    \psi:m\otimes n\mapsto \psi(m)n\,.
  \end{align}
\end{prop}

Now consider the complex algebras \(\mathbb{C}[z,z^{-1},y,y^{-1},(z-y)^{-1}]\) and
\(\mathcal{R}=\mathbb{C}[y,y^{-1}]\). An element \(f(z,y)\in \mathbb{C}[z,z^{-1},y,y^{-1},(z-y)^{-1}]\)
can be thought of as a rational function on \(\mathbb{P}\times\mathbb{P}\). We consider
three domains.
\begin{enumerate}
\item \(U_0=\{(z,y)\in \mathbb{P}\times\mathbb{P}||y|>|z|>0\}\),
\item \(U_y=\{(z,y)\in \mathbb{P}\times\mathbb{P}||y|>|z-y|>0\}\),
\item \(U_\infty=\{(z,y)\in \mathbb{P}\times\mathbb{P}||z|>|y|>0\}\),
\end{enumerate}
and define local coordinates \(\xi_a,y\) on \(U_a\), \(a=0,y,\infty\) by defining
\(\xi_0=z\), \(\xi_y=z-y\) and \(\xi_\infty=z\). Then we can define the algebra homomorphisms
\begin{align}
  j_a^L:\mathbb{C}[z,z^{-1},y,y^{-1},(z-y)^{-1}]\rightarrow \mathcal{R}((\xi_a))\,,\quad a=0,y\,,
\end{align}
and
\begin{align}
  j_\infty^R:\mathbb{C}[z,z^{-1},y,y^{-1},(z-y)^{-1}]\rightarrow \mathcal{R}((\xi_\infty^{-1}))
\end{align}
by expanding \(f(z)\) on the open sets \(U_a\) by the local coordinates \((\xi_a,y)\). We
denote \(j_a^L(f)=f_a(\xi_a;y)\) for \(a=0,y\) and \(j_\infty^R(f)=f_\infty(\xi_\infty;y)\). Then
we can define the current Lie algebra \(\mathfrak{g}^\mathbb{P}_{\mathcal{R}}(V)\) over \(\mathcal{R}\) by
\begin{align}
  \mathfrak{g}^\mathbb{P}_{\mathcal{R}}(V)=\frac{\bigoplus_{h=0}^{\infty}V[h]\otimes\mathbb{C}[z,z^{-1},y,y^{-1},(z-y)^{-1}]\d z^{1-h}}{
    \nabla(\bigoplus_{h=0}^{\infty}V[h]\otimes\mathbb{C}[z,z^{-1},y,y^{-1},(z-y)^{-1}]\d z^{-h})}\,,
\end{align}
where \(\nabla\) is defined as in \eqref{eq:currentconnection}
\begin{align}
  \nabla(A\otimes f(z,y)\d z^{-h})=L_{-1}A\otimes f(z,y)\d z^{-h}+A\otimes\frac{\d f(z,y)}{\d z}\d z^{1-h}\,.
\end{align}
Then \(\mathfrak{g}^{\mathbb{P}}_{\mathcal{R}}(V)\) has the structure of a Lie Algebra given by
\begin{align}
  &[A\otimes f(z,y)\d z^{1-h_A},B\otimes g(z,y)\d z^{1-h_B}]=\\\nonumber
  &\qquad\sum_{m=0}^\infty\frac{1}{m!} A_{m-h_A+1}B\otimes\frac{\d^m f(z,y)}{\d z^m}g(z,y)\d z^{m+2-h_A-h_B}\,,
\end{align}
Let
\begin{align}
  j_a^L:\mathfrak{g}^\mathbb{P}_{\mathcal{R}}(V)&\rightarrow \mathcal{R}\otimes \mathfrak{g}^L\,,\quad a=0,y\,,\\\nonumber
  j_\infty^R:\mathfrak{g}^\mathbb{P}_{\mathcal{R}}(V)&\rightarrow \mathcal{R}\otimes \mathfrak{g}^R\,,
\end{align}
be the Lie algebra homomorphisms over \(\mathcal{R}\) defined
in the same way as in \eqref{eq:localexpansions}.

\begin{definition}
  For \(V\)-modules \(L,M,N\) a \(\hom_{\mathbb{C}}(M\otimes N,\overline{L})\)-valued holomorphic function \(\psi(y)\) on
  \(\mathbb{C}^\ast\) -- which may be multi-valued -- is called a vertex operator of type \(\binom{M}{L,N}\) if it
  satisfies the following two conditions.
  \begin{enumerate}
  \item For \(f(z)\in \mathbb{C}[z,z^{-1},y,y^{-1},(z-y)^{-1}]\), \(A\in V[h_A]\), 
    \(m\in M\) and \(n\in N\) we have
    \begin{align}
      &j_\infty^R([A\otimes f(z,y)\d z^{1-h_A}])(\psi(m;y)n)=\\\nonumber
      &\quad \psi(j_y^L([A\otimes f(z,y)\d z^{1-h_A}]m;y))n
      +\psi(m;y)j_0^L([A\otimes f(z,y)\d z^{1-h_A}])n\,.
    \end{align}
  \item For \(m\in M\) and \(n\in N\)
    \begin{align}
      \frac{\d}{\d y}\psi(m;y)n=\psi(L_{-1}m;y)n\,.
    \end{align}
  \end{enumerate}
\end{definition}
We denote by \(I_{L,N}^M\) the complex vector space of vertex operators of type
\(\binom{M}{L,N}\). By taking \(y=1\) we can define linear maps
\begin{align}
  I_{L,N}^M&\rightarrow \hom_{\mathbb{C}}(M\otimes N,\overline{L})\\\nonumber
  \psi(-;y)&\mapsto \psi(- ;1)\,.
\end{align}
Then we have the following theorem.
\begin{thm}
  The image of the map \(\psi(- ;y)\mapsto \psi(- ;1)\) is contained in the image of the injection
  \begin{align}
    \hom_{\mathcal{U}(V)}(M\fuse N,L)\rightarrow \hom_{\mathbb{C}}(M\otimes N,\overline{L})\,,
  \end{align}
  and the two images are equal.
\end{thm}
By the above vertex operator one can define \(N\)-point conformal blocks
and prove the validity of the associativity and braiding constraints as was
shown in \cite{Huang:2003za}.

We will revisit these concepts for one special case when proving the rigidity of the 
\(\mathcal{W}_p\)-module \(X_2^+\). We also make note of a slight abuse of notation we will be using.
When considering an element \(\psi\in\hom_{\mathcal{U}(V)}(M\fuse N,L)\) and \(m\in M\),
\(n\in N\) we will identify \(m\otimes n\) with \(\psi(m)n\),
when there is no chance of confusion.

\section{The abelian category \(\mathcal{W}_p\)-mod}
\label{sec:defwpmod}
In this section we briefly review the structure of \(\mathcal{W}_p\)-mod as an abelian category following \cite{Nagatomo:2009xp}.

\subsection{\boldmath General properties of \(\mathcal{W}_p\)-mod\unboldmath}

For \(p\geq2\) the VOA \(\mathcal{W}_p\) is generated by the identity \(\mathbbm{1}\), the energy momentum tensor \(T(z)\) and three weight
\(2p-1\) primary fields \(W^\epsilon(z)\), where \(\epsilon=\pm,0\) labels \(\text{sl}_2\)-charges. The central charge of the theory is
given by
\begin{align}
  c_p=1-6\frac{(p-1)^2}{p}\,.
\end{align}
It has been shown in \cite{Adamovic:2007er} that \(\mathcal{W}_p\) is \(c_2\)-cofinite.

As an abelian category \(\mathcal{W}_p\)-mod decomposes into a \(\mathbb{C}\)-linear sum of abelian subcategories
\begin{align}
  \mathcal{W}_p\text{-mod}=\bigoplus_{s=0}^p\mathcal{C}_s\,,
\end{align}
where for \(1\leq s\leq p\) the \(\mathcal{C}_s\) are full abelian subcategories of \(\mathcal{W}_p\)-mod. For \(s\neq s^\prime\) and
\(M\in \mathcal{C}_s,\ M^\prime\in\mathcal{C}_{s^\prime}\) the spaces \(\text{Ext}_{\mathcal{W}_p}^i(M,M^\prime)=0\) for any \(i\in\mathbb{Z}\), in particular
\(\hom_{\mathcal{W}_p}(M,M^\prime)=\text{Ext}^0_{\mathcal{W}_p}(M,M^\prime)=0\).
The two subcategories \(\mathcal{C}_0\) and \(\mathcal{C}_p\) are semi-simple and contain one simple
object each
\begin{align}
  X_p^+&\in \text{obj}(\mathcal{C}_p)\,,&X_p^-&\in\text{obj}(\mathcal{C}_0)\,,
\end{align}
which are projective in \(\mathcal{C}_p\) and \(\mathcal{C}_0\) respectively as well as
in \(\mathcal{W}_p\)-mod. We will therefore occasionally also denote these modules by \(P_p^\epsilon=X_p^\epsilon\).
For \(1\leq s\leq p-1\) the subcategories \(\mathcal{C}_s\)
are not semi-simple. They contain two simple objects each
\begin{align}
  X_s^+,\ X_{p-s}^-\in\text{obj}(\mathcal{C}_s)\,.
\end{align}
We denote the projective covers of \(X_s^+\) and \(X_{p-s}^-\) by \(P_s^+\) and \(P_{p-s}^-\) respectively.
They are characterised by socle series\footnote{
A socle series of a module \(M\) is a filtration of submodules \(S_1(M)\subseteq\cdots S_n(M)=M\) such that \(S_1(M)\) is the
maximal semi-simple submodule of \(M\) and \(S_i(M)/S_{i-1}(M)\) is the maximal semi-simple submodule of \(S_{i+1}/S_{i-1}(M)\).
}
of length 3
\begin{align}
  &X_s^+=S_0(P_s^+)\subset S_1(P_s^+) \subset S_2(P_s^+)=P_s^+\\\nonumber
  &X_{p-s}^-=S_0(P_{p-s}^-)\subset S_1(P_{p-s}^-) \subset S_2(P_{p-s}^-)=P_{p-s}^-\,,
\end{align}
such that
\begin{align}
  S_1(P_s^+)/S_0(P_s^+)&=2X_{p-s}^-&S_2(P_s^+)/S_1(P_s^+)&=X_{s}^+\\\nonumber
  S_1(P_{p-s}^-)/S_0(P_{p-s}^-)&=2X_{s}^+&S_2(P_{p-s}^-)/S_1(P_{p-s}^-)&=X_{p-s}^-\,.
\end{align}
Both \(P_s^+\) and \(P_{p-s}^-\) each have two occurrences of
\(X_s^+\) and \(X_{p-s}^-\) as subquotients and
therefore they have identical characters.

The simple and the projective modules of \(\mathcal{W}_p\)-mod are all self-contragredient, {\it i.e.} 
\({X_s^\epsilon}^\ast=X_s^\epsilon\) and \({P_s^\epsilon}^\ast=P_s^\epsilon\) for \(1\leq s\leq p,\ \epsilon=\pm\). In particular the vacuum representation, \(X_1^+\) which is the
tensor unit, is self contragredient. Therefore by proposition
\ref{sec:moncatprop} it follows that \((\mathcal{W}_p\mathrm{-mod},\fuse,X_1^+)\)
is weakly rigid and that for each \(\mathcal{W}_p\)-module \(M\) the weakly
rigid dual \(M^\vee\) coincides with the contragredient \(M^\ast\).

For \(1\leq s\leq p-1\) the subcategories \(\mathcal{C}_s\) also contain 6 families of 
indecomposable modules characterised by socle series of length 2. For \(d\geq 1\) these are summarised in the table below.
\begin{center}
  \begin{tabular}{c|c|c|c|c|c|c}
    &\footnotesize{\(G_{s,d}^+\)}&\footnotesize{\(G_{p-s,d}^-\)}&\footnotesize{\(H_{s,d}^+\)}&\footnotesize{\(H_{p-s,d}^-\)}
    &\footnotesize{\(I_{s,d}^+(\lambda)\)}&\footnotesize{\(I_{p-s,d}^-(\lambda)\)}\\\hline
    \footnotesize{\(S_1/S_0\)}&\footnotesize{\((d+1) X_s^+\)}&\footnotesize{\((d+1) X_{p-s}^-\)}&\footnotesize{\(d X_s^+\)}&\footnotesize{\(d X_{p-s}^-\)}
    &\footnotesize{\(d X_s^+\)}&\footnotesize{\(d X_{p-s}^-\)}\\\hline
    \footnotesize{\(S_2/S_1\)}&\footnotesize{\(d X_{p-s}^-\)}&\footnotesize{\(d X_{s}^+\)}&\footnotesize{\((d+1) X_{p-s}^-\)}&\footnotesize{\((d+1) X_{s}^+\)}
    &\footnotesize{\(d X_{p-s}^-\)}&\footnotesize{\(d X_{s}^+\)}\\\hline
  \end{tabular}
\end{center}
Note that \(I_{s,d}^+(\lambda)\) and \(I_{p-s,d}^-(\lambda)\) are not uniquely characterised by their socle series alone.
To each of the two series there corresponds a continuous family of inequivalent indecomposable modules parametrised by 
\(\lambda\in \mathbb{P}\).

The simple modules \(X_s^\pm\) can be decomposed into direct sums of simple Virasoro modules
\begin{align}
  \label{eq:virdecomp}
  X_s^+&=\bigoplus_{m=1}^\infty (2m-1)\mathcal{L}_{h_{2m-1,s}}\,,\\\nonumber
  X_s^-&=\bigoplus_{m=1}^\infty 2m\mathcal{L}_{h_{2m,s}}\,,
\end{align}
where \(\mathcal{L}_{h_{r,s}}\) is the highest weight irreducible Virasoro module of weight
\begin{align}
  h_{r,s}=\frac{1}{4p}\left((rp-s)^2-(p-1)^2\right)\,,
\end{align}
therefore the weights of \(X_s^+\) and \(X_s^-\), which we will denote by \(h_s^+\) and \(h_s^-\), are \(h_{1,s}\) and \(h_{2,s}\) respectively.

The dimension of the highest weight spaces \(X_s^+[h_s^+]\) is 1
and the dimension of the highest weight spaces \(X_s^-[h_s^-]\) is
2. We fix non-zero vectors \(u_s\in X_s^+[h_s^+]\) and we fix a basis
\(v_s^+,v_s^-\) of \(X_s^-[h_s^-]\) which satisfies the following
conditions
\begin{align}
  W_0^+v_s^+&=0,&W_0^- v_s^-&=0\,.
\end{align}
Then \(v_s^+\) and \(v_s^-\) are universally determined up to
constants.

We have the following results.
\begin{enumerate}
\item The Virasoro submodules \(\mathcal{U}(\mathcal{L})u_s\) are
  isomorphic to the irreducible Virasoro modules
  \(\mathcal{L}_{h_{1,s}}\) and the Virasoro submodules
  \(\mathcal{U}(\mathcal{L})v_s^\mu\) are isomorphic to the
  irreducible Virasoro modules \(\mathcal{L}_{h_{2,s}}\).
\item The action of the first few modes of the \(W\)-fields on the
  highest weight vector \(u_s\) of \(X_s^+,\ 1\leq s\leq p\) is
  given by
  \begin{align}
    W_{-k}^\epsilon u_s&=0\qquad \epsilon=0,\pm,\quad k< 2p-s\,.
  \end{align}
  
\item The action of the first few modes of the \(W\)-fields on the highest weight
  vectors \(v_s^\mu\) of \(X_s^-,\ 1\leq s\leq p,\ \mu=\pm\) is then
  given by
  \begin{align}\label{eq:X1Waction}
    W_{-k}^\epsilon v_s^\mu\left\{
      \begin{array}{ll}
        =0&\epsilon=\mu\\
        \in \mathcal{U}(\mathcal{L})v_s^\mu&\epsilon=0\\
        \in \mathcal{U}(\mathcal{L})v_s^{-\mu}&\epsilon=-\mu
      \end{array}
    \right.\qquad \epsilon=\pm,0\quad k<3p-s\,.
  \end{align}
\end{enumerate}

After defining \(\mathcal{W}_p\)-mod, we now turn to \(A_0(\mathcal{W}_p)\)-mod -- the category of
zero mode quotients of \(\mathcal{W}_p\)-modules.
We define elements of \(A_0(\mathcal{W}_p)\)-mod, in the following way
\begin{align}
  \overline{X}_s^\epsilon=\mathcal{A}_0(X_s^\epsilon)\,,\quad 1\leq s\leq p,\ \epsilon=\pm\,.
\end{align}
Then the \(\overline{X}_s^\epsilon\) are simple objects in \(A_0(\mathcal{W}_p)\)-mod and any
simple object of \(A_0(\mathcal{W}_p)\)-mod is isomorphic to one of the objects \(\overline{X}_s^\epsilon\).

Just like \(\mathcal{W}_p\)-mod, \(A_0(\mathcal{W}_p)\)-mod also decomposes into a \(\mathbb{C}\)-linear
direct sum of abelian subcategories
\begin{align}
  A_0(\mathcal{W}_p)\text{-mod}=\bigoplus_{s=1,\epsilon=\pm}^p\overline{\mathcal{C}}_s^{\,\epsilon}.
\end{align}
The subcategories \(\overline{\mathcal{C}}_p^{\,+}\) and \(\overline{\mathcal{C}}_s^{\,-}, 1\leq s\leq p\) are semi-simple and
each contain one simple module
\begin{align}
  \overline{X}_p^{\,+}&\in\text{obj}(\overline{\mathcal{C}}_p^{\,+})\,,
  &\overline{X}_s^{\,-}&\in\text{obj}(\overline{\mathcal{C}}_s^{\,-})\,.
\end{align}
For \(1\leq s\leq p-1\) the subcategories \(\overline{\mathcal{C}}_s^{\,+}\) are not semi-simple. In addition to one
simple module
\begin{align}
  \overline{X}_s^{\,\epsilon},\in\text{obj}(\overline{\mathcal{C}}_s^{\,\epsilon})\,,
\end{align}
they also contain
and one reducible but indecomposable module \(\widetilde X_s^{\,+}\) that is the projective cover of \(\overline{X}_s^{\,+}\) and satisfies the
exact sequence
\begin{align}
  0\longrightarrow \overline{X}_s^{\,+}\longrightarrow \widetilde{X}_s^{\,+}\longrightarrow \overline{X}_s^{\,+}\longrightarrow 0\,.
\end{align}

The image of the indecomposable \(\mathcal{W}_p\)-modules in \(A_0(\mathcal{W}_p)\)-mod is
\begin{center}
  \begin{tabular}{lll}
    \footnotesize{\(\mathcal{A}_0(X_s^\pm)=\overline{X}_s^\pm\)}&\footnotesize{\(\mathcal{A}_0(I_{s,d}^+(\lambda))=d\overline{X}_{s}^+\)}&
    \footnotesize{\(\mathcal{A}_0(G_{s,d}^-)=(d+1)\overline{X}_s^-\oplus d \overline{X}_{p-s}^+\)}\\\nonumber
    \footnotesize{\(\mathcal{A}_0(G_{s,d}^+)=(d+1)\overline{X}_s^+\)}&\footnotesize{\(\mathcal{A}_0(P_s^+)=\widetilde{X}_s^+\)}&
    \footnotesize{\(\mathcal{A}_0(H_{s,d}^-)=d\overline{X}_s^-\oplus (d+1)\overline{X}_{p-s}^+\)}\\\nonumber
    \footnotesize{\(\mathcal{A}_0(H_{s,d}^+)=d\overline{X}_s^+\)}&\footnotesize{\(\mathcal{A}_0(P_s^-)=\overline{X}_s^-\oplus 2\overline{X}_{p-s}^+\)}&
    \footnotesize{\(\mathcal{A}_0(I_{s,d}^-(\lambda))=d\overline{X}_s^-\oplus d\overline{X}_{p-s}^+\)}\,.
  \end{tabular}
\end{center}
As one can see from the above table the indecomposable structure of \(A_0(\mathcal{W}_p)\) is much simpler than that of
\(\mathcal{W}_p\)-mod as only the images of \(P_s^+\) are non-semi-simple.

The detailed \(\mathcal{W}_p\)-module structure of \(X_1^-\) and \(X_2^+\) is crucial to calculations. We have
the following results.
\begin{prop}\label{sec:zeromodesubspaces}
  As complex vector spaces the \(\mathcal{A}_0\) quotients of simple 
  \(\mathcal{W}_p\)-modules satisfy:
  \begin{enumerate}
  \item For \(1\leq s \leq p\), \(\dim\mathcal{A}_0(X_s^+)=1\) and the space is spanned by the equivalence class represented by \(u_s\)
    and therefore has conformal weight \(h_s^+\).
  \item For \(1\leq s \leq p\), we fix the zero mode subspace \((X_s^+)^0\) to be spanned by the highest weight vector \(u_s\).
  \item For \(1\leq s\leq p\), \(\dim\mathcal{A}_0(X_s^-)=2\) and the space is spanned by the equivalence classes represented by the 
    two highest weight vectors \(v_s^\epsilon,\ \epsilon=\pm\) and therefore has conformal weight \(h_s^-\). 
  \item For \(1\leq s \leq p\), we fix the zero mode subspace \((X_s^-)^0\) to be spanned by the two highest weight vectors 
    \(v_s^\epsilon,\ \epsilon=\pm\).
  \end{enumerate}
\end{prop}

\begin{prop}\label{sec:X1special}
  The two copies of \(\mathcal{L}_{h_{2,1}}\) in the simple module \(X_1^-\)
  each contain
  a null vector at level 2
  \begin{align}\label{eq:X1null}
    (L_{-1}^2-pL_{-2})v_1^\mu=0\,,\quad \mu=\pm\,,
  \end{align}
  and a well defined choice for the special subspace \((X_1^-)^s\) is given  by
  \begin{align}
    (X_1^-)^s=\bigoplus_{j=0}^1\bigoplus_{\mu=\pm}\mathbb{C}L_{-1}^jv_1^\mu\,.
  \end{align}
\end{prop}

\begin{prop}\label{sec:X2special}
  The Virasoro submodule \(\mathcal{L}_{h_{1,2}}\) of the simple module \(X_2^+\) contains a null vector
  at level 2
  \begin{align}
    (L_{-1}^2-\tfrac1pL_{-2})u_2=0
  \end{align}
  and a well defined choice for the special subspace \((X_2^+)^s\) is given by
  \begin{align}
    (X_2^+)^s=\bigoplus_{j=0}^1\mathbb{C}L_{-1}^ju_2\oplus\bigoplus_{\epsilon=0,\pm}\mathbb{C}W_{-2p+2}^\epsilon u_2\,.
  \end{align}
\end{prop}

\subsection{The free field realisation of \(\mathcal{W}_p\)}

One can explicitly construct \(\mathcal{W}_p\) as a subVOA of a free field VOA \(V_L\) on a lattice by the method of screening operators.
The free field VOA is constructed by means of the Heisenberg algebra
\begin{align}
  \mathfrak{a}=\mathbb{C}\mathbbm{1}\oplus\bigoplus_{n\in\mathbb{Z}}\mathbb{C}a_n\,
\end{align}
as well as an operator \(\hat a\),
satisfying the commutation relations
\begin{align}
  [a_m,\hat a]&=\delta_{m,0}\,,&[a_m,a_n]=m\delta_{m,-n}\,.
\end{align}
The Heisenberg algebra acts on Fock spaces \(\mathcal{F}^\lambda\) generated by a state \(|\lambda\rangle,\ \lambda\in\mathbb{C}\)
\begin{align}
  a_m|\lambda\rangle=\lambda\delta_{m,0}|\lambda\rangle\,,\quad m\geq 0\,.
\end{align}
For the free field VOA \(V_L\) we restrict the charges \(\lambda\) of \(\mathcal{F}^\lambda\)
to a rescaled \(A_1\) root lattice \(L\) and its dual \(L^\vee \)
\begin{align}
  L&=\mathbb{Z}\alpha_+\,,&L^\vee&=\mathbb{Z}\frac{\alpha_-}{2}\,,
\end{align}
where \(\alpha_+=\sqrt{2p}\) and \(\alpha_-=-\sqrt{\tfrac2p}\). 
The theory contains a single free bosonic field
\begin{align}
  \phi(z)=\hat a + a_0\log z +\sum_{n\neq 0}\frac{a_n}{-n}z^{-n}
\end{align}
that satisfies the OPE
\begin{align}
  \phi(z)\phi(w)\sim \log(z-w)\,.
\end{align}
The energy momentum tensor is given by
\begin{align}
  T(z)=\frac12:(\partial\phi(z))^2:+\frac{\alpha_++\alpha_-}{2}\partial\phi(z)\,,
\end{align}
where \(:\ :\) indicates normal ordering, {\it i.e.} arranging the Heisenberg operators in ascending
order from left to right according to their index with \(\hat a\) on the very left.
Calculating the OPE of \(T\) with itself, one reproduces the
the central charge
\begin{align}
  c_p=1-6\frac{(p-1)^2}{p}
\end{align}
of \(\mathcal{W}_p\). The primary fields are given by
\begin{align}
  V_\mu(z)=:e^{\mu\phi(z)}:\,,
\end{align}
where \(\mu\in L^\vee\) and the weight of \(V_\mu(z)\) is
\begin{align}
  h_\mu=\frac12\mu(\mu-(\alpha_++\alpha_-))\,.
\end{align}
The OPE of two primary fields is given by
\begin{align}
  V_\mu(z)V_\nu(w)=(z-w)^{\mu\cdot\nu}:V_\mu(z)V_\nu(w):\,.
\end{align}

The VOA \(V_L\) contains the fields \(\mathbbm{1},\ T(z)\) and \(V_\mu(z)\) for \(\mu\in L\) but not 
\(V_{\nu}(z)\) for \(\nu\in L^\vee\setminus L\). The representation category \(V_L\)-mod is semi-simple with
\(2p\) simple modules \(\mathcal{V}_{[\lambda]},\ [\lambda]\in L^\vee/L\).
For later calculations it will prove useful to parametrise the classes \([\lambda]\in L^\vee/L\) by
\begin{align}
  [r,s]=:[\alpha_{r,s}]=\left[\frac{1-r}{2}\alpha_++\frac{1-s}{2}\alpha_-\right]\,\quad r,s\in \mathbb{Z}\,,
\end{align}
where
\begin{align}
  [r+1,s+p]=[r,s]\,.
\end{align}
The \(V_L\)-modules decompose into infinite sums of Fock spaces
\begin{align}
  \mathcal{V}_{[\alpha_{r,s}]}=\bigoplus_{n\in\mathbb{Z}}\mathcal{F}^{\alpha_{r+2n,s}}\,.
\end{align}

The \(V_L\)-theory contains two weight 1 primary fields that can be used as screening operators
\begin{align}
  Q_+(z)&=V_{\alpha_+}(z)\,,&Q_-(z)&=V_{\alpha_-}(z)\,.
\end{align}
The \(\mathcal{W}_p\) VOA is realised by screening with \(Q_-(z)\)
\begin{align}
  \mathcal{W}_p=\ker \left(\oint \d z\,Q_-(z):\mathcal{V}_{[1,1]}\rightarrow \mathcal{V}_{[1,-1]}\right)\,.
\end{align}

As \(\mathcal{W}_p\)-modules the simple \(V_L\)-modules \(\mathcal{V}_{[1,p]}\) and \(\mathcal{V}_{[2,p]}\) are isomorphic to
\(X_p^+\) and \(X_p^-\) respectively. The remaining \(2p-2\) simple \(V_L\)-modules are reducible as
\(\mathcal{W}_p\)-modules and form Felder complexes \cite{Fel:1989}
\begin{align}
  \begin{tikzpicture}[baseline=(1)]
    \path
    (0,0) node (1) {\(\cdots\)}
    ++(2,0) node (2) {\(\mathcal{V}_{[\alpha_{1,s}]}\)}
    ++(3,0) node (3) {\(\mathcal{V}_{[\alpha_{2,p-s}]}\)}
    ++(3,0) node (4) {\(\mathcal{V}_{[\alpha_{1,s}]}\)}
    ++(2,0) node (5) {\(\cdots\)\ .};
    \draw[->,>=latex] (1) -- (2);
    \draw[->,>=latex] (2) -- (3) node[above,midway] {\(Q_-^{(s)}\)};
    \draw[->,>=latex] (3) -- (4) node[above,midway] {\(Q_-^{(p-s)}\)};
    \draw[->,>=latex] (4) -- (5);
  \end{tikzpicture} 
\end{align}
where \(Q_-^{(a)}\) is the zero mode of a rather complicated \(a\)-fold product of \(Q_-(z)\)
whose details need not concern us here \cite{Tsuchiya:1986}. The simple \(\mathcal{W}_p\)-modules \(X_s^\pm\) for
\(1\leq s\leq p-1\) are equivalent to the kernels and images of \(Q^{(a)}\)
\begin{align}
  &X_s^+=\ker \left(Q_-^{(s)}:\, \mathcal{V}_{[1,s]}\rightarrow\mathcal{V}_{[2,p-s]}\right)
  =\im\left( Q_-^{(p-s)}:\,\mathcal{V}_{[2,p-s]}\rightarrow \mathcal{V}{[1,s]}\right)\\\nonumber
  &X_{p-s}^-=\ker \left(Q_-^{(p-s)}:\, \mathcal{V}_{[2,p-s]}\rightarrow\mathcal{V}_{[1,s]}\right)
  =\im\left( Q_-^{(s)}:\, \mathcal{V}_{[1,s]}\rightarrow \mathcal{V}{[2,p-s]}\right)\,.
\end{align}
\begin{prop}\label{sec:WtoV}
  The screening operator maps \(Q^{(a)}\) induce surjective \(\mathcal{W}_p\)-module maps
  \begin{align}
    \mathcal{V}_{[1,s]}&\rightarrow X_{p-s}^{-}\,,& \mathcal{V}_{[2,s]}&\rightarrow X_{p-s}^{+}\,,\\\nonumber
    |\alpha_{-1,s}\rangle &\mapsto v^+_{p-s}& |\alpha_{0,s}\rangle &\mapsto u_{p-s}
  \end{align}
  for \(1\leq s\leq p-1\) as well as injective \(\mathcal{W}_p\)-module maps
  \begin{align}
    X_s^{+}&\rightarrow \mathcal{V}_{[1,s]}\,,& X_s^{-}&\rightarrow \mathcal{V}_{[2,s]}\,,\\\nonumber
    u_s&\mapsto |\alpha_{1,s}\rangle& v_s^-&\mapsto |\alpha_{2,s}\rangle
  \end{align}
  for \(1\leq s\leq p\).
\end{prop}

\subsection{The \(V_L\) fusion product}
\label{sec:vlfusion}

We recall some well known facts about the \(V_L\)-mod fusion product that will be relevant
to our calculations below.
The tuple \((V_L\mathrm{-mod},\fuse_{V_L},\mathcal{V}_{[1,1]})\) defines a braided monoidal category,
{\it i.e.} there is a well defined fusion product of \(V_L\)-modules
\begin{align}
  \mathcal{V}_{[s_1,r_1]}\fuse_{V_L}\mathcal{V}_{[s_2,r_2]}=\mathcal{V}_{[s_1+s_2-1,r_1+r_2-1]}\,.
\end{align}

The free field VOA \(V_L\) contains another subVOA
\((\mathcal{F}^{\alpha_{1,1}},|0\rangle,T,Y)\)
 other than \(\mathcal{W}_p\) called the Heisenberg VOA.\footnote{This is the only appearance of a non-\(c_2\)-cofinite VOA in this paper.}
 The current algebra \(\mathcal{U}(\mathcal{F}^{\alpha_{1,1}})\)
is given by the universal enveloping algebra \(\mathcal{U}(\mathfrak{a})\) of the Heisenberg algebra, while the
simple objects of \(\mathcal{F}^{\alpha_{1,1}}\) are given by 
\(\mathcal{F}^{\lambda},\lambda\in\mathbb{C}\)
\begin{prop}\label{sec:Vfusion}
  For \(\lambda_1,\lambda_2\in\mathbb{C}\) the fusion product in \(\mathcal{F}^{\alpha_{1,1}}\)-mod
  is given by
  \begin{align}
    \mathcal{F}^{\lambda_1}\fuse_{\mathcal{F}^{\alpha_{1,1}}}\mathcal{F}^{\lambda_2}
    &\stackrel{\cong}{\rightarrow}
    \mathcal{F}^{\lambda_1+\lambda_2}\\\nonumber
    |\lambda_1\rangle\otimes|\lambda_2\rangle&\mapsto |\lambda_1+\lambda_2\rangle
  \end{align}
  and for \((r_1,s_1),(r_2,s_2)\in\mathbb{Z}^2\)
  the following diagram commutes
\begin{align}
  \begin{tikzpicture}[baseline=(base)]
    \path
    (0,0) node (1) {\(\mathcal{F}^{\alpha_{r_1,s_1}}\fuse_{\mathcal{F}^{\alpha_{1,1}}}\mathcal{F}^{\alpha_{r_2,s_2}}\)}
    ++(0,-2) node (2) {\(\mathcal{V}_[r_1,s_1]\fuse_{V_L}\mathcal{V}_{[r_2,s_2]}\)}
    ++(5,2) node (3) {\(\mathcal{F}^{\alpha_{r_1+r_2-1,s_1+s_2-1}}\)}
    ++(0,-2) node (4) {\(\mathcal{V}_{[r_1+r_2-1,s_1+s_2-1]}\)};
    \node (0,-1.5) (base) {};
    \draw[->,>=latex] (1) -- (2) node[right,midway] {};
    \draw[->,>=latex] (2) -- (4) node[right,midway] {};
    \draw[->,>=latex] (1) -- (3) node[above,midway] {};
    \draw[->,>=latex] (3) -- (4) node[above,midway] {};
  \end{tikzpicture}
\end{align}
where 
the vertical arrows are injective \(\mathcal{F}^{\alpha_{1,1}}\)-module maps and
the horizontal arrows are a \(\mathcal{F}^{\alpha_{1,1}}\) and a \(V_L\)-isomorphism
respectively.
\end{prop}

\section{\boldmath Computing certain fusion products in \(\mathcal{W}_p\)-mod \unboldmath}
\label{sec:calcX1andX2}
In this section we apply the methods explained above to analyse the 
monoidal structure of \(\mathcal{W}_p\)-mod.

\subsection{The fusion rules and rigidity of \(X_1^-\)}

In this section we analyse the fusion products of \(X_1^-\) with simple modules and prove the rigidity
of \(X_1^-\).
\begin{thm}\label{sec:X1msquare}
  \begin{enumerate}
  \item The fusion square of \(X_1^-\) is
    \begin{align}
      X_1^-\fuse X_1^-=X_1^+\,.
    \end{align}
  \item \(X_1^-\) is rigid and self-dual.
  \end{enumerate}
\end{thm}
\begin{proof}[Sketch of the proof]
  We prove the the theorem in three steps.
  \begin{enumerate}
  \item We prove that there exists a surjection of \(A_0(\mathcal{W}_p)\)-modules
    \begin{align}
      \mathcal{A}_0(X_1^+)\rightarrow \mathcal{A}_0(X_1^-\fuse X_1^-)\,.
    \end{align}
  \item We prove that \(\dim \mathcal{A}_1(X_1^-\fuse X_1^-)[1]=0\).
  \item We prove the existence of a non-trivial \(\mathcal{W}_p\)-module map
    \begin{align}
      X_1^-\fuse X_1^-\rightarrow \mathcal{V}_{[2,p-1]}\,.
    \end{align}
  \end{enumerate}
  Step 1 implies that \(X_1^-\fuse X_1^-\) is a (possibly trivial) highest weight module 
  generated by a state of conformal weight 0. Since \(h_{p-1}^-=1\) step 2 excludes the possibility
  of \(X_{p-1}^-\) being a submodule of \(X_1^-\fuse X_1^-\). Step 3 implies that \(X_1^-\fuse X_1^-\)
  is non-trivial and since the only non-trivial submodule of \(\mathcal{V}_{[2,p-1]}\), generated by a
  state of conformal weight 0 is \(X_1^+\), it follows that \(X_1^-\fuse X_1^-=X_1^+\).
  
  The rigidity of \(X_1^-\) follows by choosing
  \begin{align}
    \text{id}_{X_1^+}&=e_{X_1^-}: X_1^-\fuse X_1^-\rightarrow X_1^+\\\nonumber
    \text{id}_{X_1^+}&=i_{X_1^-}: X_1^+\rightarrow X_1^-\fuse X_1^-
  \end{align}
  and the fact that therefore all the maps appearing in definition \ref{sec:defrigid}
  are isomorphisms.
\end{proof}

\begin{proof}[Proof of step 1]
  As in proposition \ref{sec:X1special} we choose
  \begin{align}
    (X_1^-)^s=\bigoplus_{j=0}^1\bigoplus_{\epsilon=\pm}\mathbb{C}L_{-1}^jv_1^\epsilon
  \end{align}
  and as in proposition \ref{sec:zeromodesubspaces} we choose
  \begin{align}
    (X_1^-)^0=\bigoplus_{\epsilon=\pm}\mathbb{C}v_1^\epsilon\,.
  \end{align}

  Using the formulae \eqref{eq:vircomult} as well as the null vector in proposition \ref{sec:X1special}
  we can compute the action of \(L_0\) on the classes represented by the elements of \((X_1^-)^s\otimes (X_1^-)^0\).
  \begin{align}
    [v_1^{\epsilon_1}\otimes v_1^{\epsilon_2}]&\mapsto 2h_1^-[v_1^{\epsilon_1}\otimes v_1^{\epsilon_2}]+[L_{-1}v_1^{\epsilon_1}\otimes v_1^{\epsilon_2}]\\\nonumber
    [L_{-1}v_1^{\epsilon_1}\otimes v_1^{\epsilon_2}]&\mapsto (2h_1^-+1)[v_1^{\epsilon_1}\otimes v_1^{\epsilon_2}]
    +[L_{-1}^2v_1^{\epsilon_1}\otimes v_1^{\epsilon_2}]\\\nonumber
    &\phantom{\mapsto}\,=ph_1^-[v_1^{\epsilon_1}\otimes v_1^{\epsilon_2}]+(2h_1^-+1-p)[L_{-1}v_{\epsilon_1}\otimes v_1^{\epsilon_2}]\,.
  \end{align}
  Thus for each pair \(\epsilon_1,\epsilon_2\) we can represent \(L_0\) by
  \begin{align}
    L_0\cong\left(
      \begin{array}{cc}
        2h_1^-&ph_1^-\\
        1&2h_1^-+1-p
      \end{array}
    \right)
  \end{align}
  on the basis \(L_{-1}^jv_1^{\epsilon_1}\otimes v_1^{\epsilon_2},\ j=0,1\). The eigenvalues of this matrix are \(h_1^+=0\) and \(2p-1\).

  Next we will determine a lower bound on the dimension of the kernel of the surjection
  \begin{align}\label{eq:X1squaresurj}
    (X_1^-)^s\otimes (X_1^-)^0\rightarrow \mathcal{A}_0(X_1^-\fuse X_1^-)\,.
  \end{align}
  Because \(A_0(\mathcal{W}_p)\)-mod does not contain any module with eigenvalue \(2p-1\) eigenvectors, the eigenvectors
  \begin{align}
    v_1^{\epsilon_1}\otimes v_1^{\epsilon_2}+\tfrac{2}{3p-2}L_{-1}v_1^{\epsilon_1}\otimes v_1^{\epsilon_2}\,,
  \end{align}
  corresponding to the eigenvalue \(2p-1\), must lie in the kernel of the surjection \eqref{eq:X1squaresurj}.
  
  From formula \eqref{eq:X1Waction} illustrating the action of \(W\)-field modes on \(X_1^-\) we know
  \begin{align}
    L_{-1}v_1^\epsilon=C_\epsilon\cdot W^\epsilon_{-1}v_1^{-\epsilon}\qquad \epsilon=\pm
  \end{align}
  for some constant \(C_\epsilon\). 
  This implies that \(L_{-1}v_1^\epsilon\otimes v_1^\epsilon,\ \epsilon=\pm\) lies in the kernel of \eqref{eq:X1squaresurj}, because
  \begin{align}
    &[L_{-1}v_1^{\epsilon}\otimes v_1^{\epsilon}]=C_\epsilon\cdot [W^{\epsilon}_{-1}v_1^{-\epsilon}\otimes v_1^{\epsilon}]\\\nonumber
    &\qquad= C_\epsilon\cdot\sum_{j=0}^{2p-3}\binom{2p-3}{j}(-1)^j[v_1^{-\epsilon}\otimes W^{\epsilon}_{j-(2p-2)}v_1^{\epsilon}]=0\,.
  \end{align}
  Finally 
  \begin{align}
    &[L_{-1}v_1^{\epsilon}\otimes v_1^{-\epsilon}]=C_\epsilon\cdot [W^{\epsilon}_{-1}v_1^{-\epsilon}\otimes v_1^{-\epsilon}]\\\nonumber
    &\qquad= C_\epsilon\cdot\sum_{j=0}^{2p-3}\binom{2p-3}{j}(-1)^j[v_1^{-\epsilon}\otimes W^{\epsilon}_{j-(2p-2)}v_1^{-\epsilon}]\\\nonumber
    &\qquad=A_\epsilon [L_{-1}v_1^{-\epsilon}\otimes v_1^{\epsilon}]+B_\epsilon[v_1^{-\epsilon}\otimes v_1^{\epsilon}]\,,
  \end{align}
  for some constants \(A_\epsilon\) and \(B_\epsilon\), since by the action of the \(W\)-modes \eqref{eq:X1Waction} 
  \(W^\epsilon_{j-(2p-2)}v_1^{-\epsilon}\in\mathcal{L}_{h_{2,1}}^\epsilon\).
  This implies that some non-trivial linear combination of \([L_{-1}v_1^{+}\otimes v_1^{-}]\) and 
  \([L_{-1}v_1^{-}\otimes v_1^{+}]\) lies in the kernel of \eqref{eq:X1squaresurj}.
  Therefore the kernel of \eqref{eq:X1squaresurj} is at least 7 dimensional
  and \(\mathcal{A}_0(X_1^-\fuse X_1^-)\) is at most one dimensional. If \(\mathcal{A}_0(X_1^-\fuse X_1^-)\) is indeed non-trivial, then the eigenvalue
  of \(L_0\) is 0.
\end{proof}

\begin{proof}[Proof of step 2]  
  We know that the image of the action of the \(W\)-field modes \(W^\mu_{-k}\) on the
  highest weight vectors \(v_1^\epsilon\) lies in the Virasoro submodules generated by \(v_1^+\) and \(v_1^-\)
  for \(k<3p-1\). Therefore
  a spanning set of representatives for \(\mathcal{A}_1(X_1^-\fuse X_1^-)\) can be chosen from
  Virasoro descendants of \(v_1^{\epsilon_1}\otimes v_1^{\epsilon_2},\ \epsilon_1=\pm,\epsilon_2=\pm\). Also since the relations \eqref{eq:vircomult} for
  Virasoro modes still hold for \(n\geq 2\), we can restrict the spanning set of representatives for
  \(\mathcal{A}_1(X_1^-\fuse X_1^-)\) to \(L_{-1}\)-descendants of \(v_1^{\epsilon_1}\otimes v_1^{\epsilon_2}\). Finally because of the null vector \eqref{eq:X1null}
  at level 2 of \(X_1^-\) we have the following surjection of complex vector spaces
  \begin{align}
    \bigoplus_{\substack{i=0\\j=0}}^1\bigoplus_{\substack{\epsilon_1=\pm\\\epsilon_2=\pm}}
    \mathbb{C}[L_{-1}^iv_1^{\epsilon_1}\otimes L_{-1}^jv_1^{\epsilon_2}]\rightarrow 
    \mathcal{A}_1(X_1^-\fuse X_1^-)\,.
  \end{align}
  Because the image of the canonical Lie algebra homomorphism
  \begin{align}
    \mathfrak{g}(\mathcal{W}_p)\rightarrow \mathcal{U}(\mathcal{W}_p)
  \end{align}
  is dense, we know that the image of \(L_{-1}^2\) lies in \(\mathfrak{g}_2(\mathcal{W}_p)(X_1^-\fuse X_1^-)\)
  and that \(L_{-1}^2\) therefore acts trivially on \(\mathcal{A}_1(X_1^-\fuse X_1^-)\) even if
  \(L_{-1}\) does not. This implies the relation
  \begin{align}
    &(j_{1,0}([T\otimes 1\cdot\d z^{-1}]))^2[v_1^{\epsilon_1}\otimes v_1^{\epsilon_2}]\\\nonumber
    &\qquad=[L_{-1}v_1^{\epsilon_1}\otimes v_1^{\epsilon_2}]+ 2 [L_{-1}v_1^{\epsilon_1}\otimes L_{-1}v_1^{\epsilon_2}]
    +[v_1^{\epsilon_1}\otimes L_{-1}v_1^{\epsilon_2}]=0
  \end{align}
  We therefore take \([v_1^{\epsilon_1}\otimes v_1^{\epsilon_2}],\ [L_{-1}v_1^{\epsilon_1}\otimes v_1^{\epsilon_2}]\) and \([v_1^{\epsilon_1}\otimes L_{-1}v_1^{\epsilon_2}]\) as a spanning set for \(\mathcal{A}_1(X_1^{-1}\fuse X_1^{-1})\)
  and compute the action of \(L_0\)
  \begin{align}
    [v_1^{\epsilon_1}\otimes v_1^{\epsilon_2}]&\mapsto\left(\frac32p-1\right)[v_1^{\epsilon_1}\otimes v_1^{\epsilon_2}]+[L_{-1}v_1^{\epsilon_1}\otimes v_1^{\epsilon_2}]\\\nonumber
    [L_{-1}v_1^{\epsilon_1}\otimes v_1^{\epsilon_2}]&\mapsto\left(\frac32p-1\right)\frac{p}{2}[v_1^{\epsilon_1}\otimes v_1^{\epsilon_2}]+\frac32p[L_{-1}v_1^{\epsilon_1}\otimes v_1^{\epsilon_2}]
    +p[v_1^{\epsilon_1}\otimes L_{-1}v_1^{\epsilon_2}]\\\nonumber
    [v_1^{\epsilon_1}\otimes L_{-1}v_1^{\epsilon_2}]&\mapsto-\left(\frac32p-1\right)\frac{p}{2}[v_1^{\epsilon_1}\otimes v_1^{\epsilon_2}]+\frac{p}{2}[L_{-1}v_1^{\epsilon_1}\otimes v_1^{\epsilon_2}]
    +p[v_1^{\epsilon_1}\otimes L_{-1}v_1^{\epsilon_2}]\,.
  \end{align}
  As a matrix \(L_0\) is represented by
  \begin{align}
    \left(
      \begin{array}{ccc}
        \left(\frac32p-1\right)&\left(\frac32p-1\right)\frac{p}{2}&-\left(\frac32p-1\right)\frac{p}{2}\\
          1&\frac32p&\frac{p}{2}\\
          0&p&p
      \end{array}
    \right)
  \end{align}
  on the basis \(v_1^{\epsilon_1}\otimes v_1^{\epsilon_2},\ L_{-1}v_1^{\epsilon_1}\otimes v_1^{\epsilon_2}\)
  and \(v_1^{\epsilon_1}\otimes L_{-1}v_1^{\epsilon_2}\)
  and the eigenvalues of this matrix are \(0,\ 2p-1\) and \(2p\), none of which are \(h_{p-1}^-=1\).
\end{proof}

\begin{proof}[Proof of step 3]
  According to proposition \ref{sec:WtoV} there exists an injective \(\mathcal{W}_p\)-module
  map
  \begin{align}
    X_1^-&\rightarrow \mathcal{V}_{[2,1]}\\\nonumber
    v^-_1&\mapsto|\alpha_{2,1}\rangle\,.
  \end{align}
  By this map and proposition \ref{sec:subfusion}
  there exits a non-trivial \(\mathcal{W}_p\)-module map
  \begin{align}
    X_1^-\fuse X_1^-&\rightarrow \mathcal{V}_{[2,1]}\fuse_{V_L}\mathcal{V}_{[2,1]}\\\nonumber
    v_1^-\otimes v_1^-&\mapsto |\alpha_{2,1}\rangle\otimes|\alpha_{2,1}\rangle\,.
  \end{align}
  By proposition \ref{sec:Vfusion} there exits a \(V_L\)-module isomorphism
  \begin{align}
    \mathcal{V}_{[2,1]}\fuse_{V_L}\mathcal{V}_{[2,1]}&\rightarrow \mathcal{V}_{[3,1]}\cong\mathcal{V}_{[1,1]}\\\nonumber
    |\alpha_{2,1}\rangle\otimes|\alpha_{2,1}\rangle&\mapsto |\alpha_{3,1}\rangle\,.
  \end{align}
  By composing these two maps we have constructed a non-trivial \(\mathcal{W}_p\)-module map
  \begin{align}
    X_1^-\fuse X_1^-\rightarrow \mathcal{V}_{[1,1]}\,.
  \end{align}
\end{proof}

\begin{thm}\label{sec:X1fusion}
  The fusion rules of \(X_1^-\) with simple modules is given by
  \begin{align}
    X_1^-\fuse X_s^\epsilon=X_s^{-\epsilon}\quad 1\leq s\leq p,\ \epsilon=\pm\,.
  \end{align}
\end{thm}
\begin{proof}[Sketch of proof]
  We prove the theorem in two steps
  \begin{enumerate}
  \item Let \(M\) be a simple module, then \(X_1^-\fuse M\) is also simple.
  \item We prove the existence of a non-trivial \(\mathcal{W}_p\)-module map
    \begin{align}\label{eq:X1Xsintertwiner}
      X_1^-\fuse X_s^+\rightarrow \mathcal{V}_{[2,s]}\,.
    \end{align}
  \end{enumerate}
  The simplicity of \(X_1^-\fuse X_s^+\) implied by step 1 and the non-triviality of the map in step 2 implies that
  \(X_1^-\fuse X_s^+\) is a simple submodule of \(\mathcal{V}_{[2,s]}\). Therefore \(X_1^-\fuse X_s^+=X_s^-\). The theorem then follows by \(X_1^-\fuse X_1^-=X_1^+\).
\end{proof}

\begin{proof}[Proof of step 1]
  Proof by contradiction. Assume \(X_1^-\fuse M\) is not simple, then there exists an
  exact sequence
  \begin{align}
    0 \rightarrow A\rightarrow X_1^-\fuse M\rightarrow B\rightarrow 0\,,
  \end{align}
  for some non-trivial \(\mathcal{W}_p\)-modules \(A\) and \(B\). Because \(X_1^-\) is
  rigid, the sequence
  \begin{align}
    0 \rightarrow X_1^-\fuse A\rightarrow M\rightarrow X_1^-\fuse B\rightarrow 0\,,
  \end{align}
  must also be exact which is in contradiction to \(M\) being simple.
\end{proof}

\begin{proof}[Proof of step 2]
According to proposition \ref{sec:WtoV} there exist injective \(\mathcal{W}_p\)-module
maps
\begin{align}
  X_1^-&\rightarrow \mathcal{V}_{[2,1]}\\\nonumber
  v^-_1&\mapsto|\alpha_{2,1}\rangle
\end{align}
and for \(1\leq s\leq p\)
\begin{align}
  X_s^+&\rightarrow \mathcal{V}_{[1,s]}\,.\\\nonumber
  u_s&\mapsto |\alpha_{1,s}\rangle
\end{align}
By the above injective \(\mathcal{W}_p\)-module maps and proposition \ref{sec:subfusion}
there exits a non-trivial \(\mathcal{W}_p\)-module map
\begin{align}
  X_1^-\fuse X_s^+&\rightarrow \mathcal{V}_{[2,1]}\fuse_{V_L}\mathcal{V}_{[1,s]}\\\nonumber
  v_1^-\otimes u_s&\mapsto |\alpha_{2,1}\rangle\otimes|\alpha_{1,s}\rangle\,.
\end{align}
By proposition \ref{sec:Vfusion} there exits a \(V_L\)-module isomorphism
\begin{align}
  \mathcal{V}_{[2,1]}\fuse_{V_L}\mathcal{V}_{[1,s]}&\rightarrow \mathcal{V}_{[2,s]}\\\nonumber
  |\alpha_{2,1}\rangle\otimes|\alpha_{1,s}\rangle&\mapsto |\alpha_{2,s}\rangle\,.
\end{align}
By composing these two maps we have constructed a non-trivial \(\mathcal{W}_p\)-module map
\begin{align}
  X_1^-\fuse X_s^+\rightarrow \mathcal{V}_{[2,s]}\,.
\end{align}
\end{proof}

\subsection{The fusion rules and rigidity of \(X_2^+\)}

In this section we analyse the fusion products of \(X_2^+\) with simple modules and prove the rigidity
of \(X_2^+\).

\begin{thm}\label{sec:X2Xsthm}
  The \(\mathcal{W}_p\)-module \(X_2^+\) is rigid and
  he fusion rules of \(X_2^+\) with simple modules is given by
  \begin{align}
    X_2^+\fuse X_s^\epsilon=\left\{
        \begin{array}{ll}
          X_2^\epsilon&s=1\\
          X_{s-1}^\epsilon\oplus X_{s+1}^\epsilon&2\leq s\leq p-1\\
          P_{p-1}^\epsilon&s=p
        \end{array}\right.\,.
  \end{align}
\end{thm}
\begin{proof}[Sketch of proof]
  We prove the theorem in a number of steps
  \begin{enumerate}
  \item We prove the existence of surjections of \(\mathcal{A}_0\)-modules
    \begin{align}
      \left.
        \begin{array}{rr}
          s=1&\mathcal{A}_0(X_2^-)\\
          1<s<p&\mathcal{A}_0(X_{s-1}^-)\oplus \mathcal{A}_0(X_{s+1}^-)\\
          s=p&\mathcal{A}_0(P_{p-1}^-)
        \end{array}
      \right\}\rightarrow
      \mathcal{A}_0(X_2^+\fuse X_s^-)\rightarrow 0
      \,.
    \end{align}
  \item We prove the existence of non-trivial \(\mathcal{W}_p\)-module maps
    \begin{align}
      X_2^+\fuse X_s^-\rightarrow \mathcal{V}_{[2,s+1]}\,,
    \end{align}
    for \(1\leq s\leq p-1\).
  \item We prove the existence of non-trivial \(\mathcal{W}_p\)-module maps
    \begin{align}
      X_2^+\fuse \mathcal{V}_{[1,p-s]}\rightarrow X_{s-1}^-\,,
    \end{align}
    for \(2\leq s\leq p-1\).
  \item We prove the existence of a surjective \(\mathcal{W}_p\)-module map
    \begin{align}
      X_2^+\fuse X_p^-\rightarrow X_{p-1}^-\,.
    \end{align}
  \item We use the formalism outlined in section \ref{sec:intertwiners}
    to prove that \(X_2^+\) is rigid and that therefore \((X_2^+)^\vee=X_2^+\).
  \end{enumerate}

  Steps 1 through 3 prove
  \begin{align}
    X_2^+\fuse X_s^{-}=\left\{
      \begin{array}{cc}
        X_2^{-}& s=1\\
        X_{s-1}^{-}\oplus X_{s+1}^{-}&1<s<p
      \end{array}
      \right.\,.
  \end{align}
  Step 5 implies that \(X_2^+\fuse X_p^-\) is projective, since the
  product of the dual of a rigid module and a projective module is
  again projective. The
  only projective module compatible with steps 1 and 4 is \(P_{p-1}^-\),
  therefore
  \begin{align}
    X_2^+\fuse X_p^{-}=P_{p-1}^-\,.
  \end{align}
  Finally the fusion products of the theorem follow by multiplying with \(X_1^-\) and the associativity
  of the fusion product.
\end{proof}

\begin{proof}[Proof of step 1]
We choose the special subspace of \(X_2^+\) to be given by
  \begin{align}
    (X_2^+)^s=\bigoplus_{j=0}^1\mathbb{C}L_{-1}^ju_2\oplus\bigoplus_{\epsilon=0,\pm}\mathbb{C}W_{-2p+2}^\epsilon u_2\,,
  \end{align}
as in proposition \ref{sec:X2special} 
and we chose the zero mode 
subspace of \(X_s^-\) to be given by
\begin{align}
  (X_s^-)^0=\bigoplus_{\epsilon=\pm}\mathbb{C}v_s^\epsilon\,,
\end{align}
as in proposition \ref{sec:zeromodesubspaces}.
By proposition \ref{sec:Akestimate} there is a canonical surjection
\begin{align}
  (X_2^+)^s\otimes (X_1^-)^0\rightarrow \mathcal{A}_0(X_2^+\fuse X_s^-)\,.
\end{align}

We first show that the spanning set
\begin{align}
  \text{span}\{[L_{-1}^j u_2\otimes v_s^\epsilon],\ [W_{-2p+2}^\mu u_2\otimes v_s^\epsilon],\ j=0,1,\ \epsilon=\pm,\ \mu=\pm,0\}
\end{align}
is redundant. Consider
\begin{align}
  [W_{-2p+2}^\mu u_2\otimes v_s^\epsilon]=-[u_2\otimes W_{-2p+2}^\mu v_s^\epsilon]\,.
\end{align}
The difference in conformal highest weights between the Virasoro
representations
\(\mathcal{L}_{h_{4,s}}\) and \(\mathcal{L}_{h_{2,s}}\)
in the
decomposition \eqref{eq:virdecomp} of \(X_s^-\), 
is \(h_{4,s}-h_{2,s}=3p-s\). Therefore
\begin{align}
  W_{-2p+2}^\mu v_s^\epsilon\in \bigoplus_{\delta=\pm}\mathcal{U}(\mathcal{L})v_s^\delta
\end{align}
which implies that 
\([W_{-2p+2}^\mu u_2\otimes v_s^\epsilon]\) depends linearly on the
 \([L_{-1}^j u_2\otimes v_s^\delta],\ j=0,1,\ \delta=\pm\).

We therefore have a 4 dimensional spanning set for \(\mathcal{A}_0(X_2^+\fuse X_s^-)\) on
which we can compute the action of \(L_0\). 
  \begin{align}
    [u_2\otimes v_s^\pm]&\mapsto(h_2^++h_s^+)[u_2\otimes v_s^\pm]+[L_{-1}u_2\otimes v_s^\pm]\\\nonumber
    [L_{-1}u_2\otimes v_s^\pm]&\mapsto\frac{h_s^+}{p}[u_2\otimes v_s^\pm]+(h_2^++h_s^++1-\frac1p)[L_{-1}u_2\otimes v_s^\pm]\,.
  \end{align}
  We can therefore represent \(L_0\) by the matrix
  \begin{align}
    \left(
      \begin{array}{cc}
        h_2^++h_s^-&\frac{h_s^-}{p}\\
        1& h_2^++h_s^-+1-\frac1p
      \end{array}
    \right)
  \end{align}
  on the basis \(L_{-1}^j u_2\otimes v_s^\epsilon,\ j=0,1,\ \epsilon=\pm\).
  For \(2\leq s\leq p-1\) the eigenvalues of this matrix are \(h_{s-1}^-\) and \(h_{s+1}^-\) and
  for \(s=p\) the eigenvalues of the above matrix are \(h_{p-1}^-\) and \(h_{1}^+\). 
\end{proof}

\begin{proof}[Proof of step 2]
  According to proposition \ref{sec:WtoV} there exist injective \(\mathcal{W}_p\)-module
  maps
  \begin{align}
    X_2^+&\rightarrow \mathcal{V}_{[2,s]}\\\nonumber
    u_2&\mapsto|\alpha_{1,s}\rangle\,.
  \end{align}
  and for \(1\leq s< p\)
  \begin{align}
    X_s^-&\rightarrow \mathcal{V}_{[2,s]}\\\nonumber
    v_s^-&\mapsto |\alpha_{2,s}\rangle\,.
  \end{align}
  By these maps and proposition \ref{sec:subfusion}
  there exits a non-trivial \(\mathcal{W}_p\)-module map
  \begin{align}
    X_2^+\fuse X_s^-&\rightarrow \mathcal{V}_{[1,2]}\fuse_{V_L}\mathcal{V}_{[2,s]}\\\nonumber
    u_2\otimes v_s^-&\mapsto |\alpha_{1,2}\rangle\otimes|\alpha_{2,s}\rangle\,.
  \end{align}
  By proposition \ref{sec:Vfusion} there exits a \(V_L\)-module isomorphism
  \begin{align}
    \mathcal{V}_{[1,2]}\fuse_{V_L}\mathcal{V}_{[2,s]}&\rightarrow \mathcal{V}_{[2,s+1]}\\\nonumber
    |\alpha_{1,2}\rangle\otimes|\alpha_{2,s}\rangle&\mapsto |\alpha_{2,s+1}\rangle\,.
  \end{align}
  By composing these two maps we have constructed a non-trivial \(\mathcal{W}_p\)-module map
  \begin{align}
    X_2^+\fuse X_s^-\rightarrow \mathcal{V}_{[2,s+1]}\,.
  \end{align}
\end{proof}

\begin{proof}[Proof of step 3]
  Before we begin with the proof we note that the proof of step one implies the existence of a
  surjective \(\mathcal{W}_p\)-module map
  \begin{align}
    X_{s-1}^-\oplus X_{s+1}^-\rightarrow X_2^+\fuse X_s^-
  \end{align}
  for \(1<s<p\), {\it i.e.} the results for \(\mathcal{A}_0(X_2^+\fuse X_s^-)\) allow for no
  modules larger than \(X_{s-1}^-\) or \(X_{s+1}^-\). Therefore because \(X_1^-\) is rigid and the functor \(X_1^-\fuse -\)
  is exact, there exits a surjective \(\mathcal{W}_p\)-module map
  \begin{align}\label{eq:auxupperbound}
    X_{s-1}^+\oplus X_{s+1}^+\rightarrow X_2^+\fuse X_s^+\,.
  \end{align}

  According to proposition \ref{sec:WtoV} there exists an injective \(\mathcal{W}_p\)-module map
  \begin{align}
    X_2^+&\rightarrow \mathcal{V}_{[2,s]}\\\nonumber
    u_2&\mapsto|\alpha_{1,s}\rangle\,.
  \end{align}
  By this map and proposition \ref{sec:subfusion}
  there exits a non-trivial \(\mathcal{W}_p\)-module map
  \begin{align}
    X_2^+\fuse \mathcal{V}_{[1,p-s]}&\rightarrow \mathcal{V}_{[1,2]}\fuse_{V_L}\mathcal{V}_{[1,p-(s-1)]}\\\nonumber
    u_2\otimes |\alpha_{-1,p-s}\rangle&\mapsto |\alpha_{1,2}\rangle\otimes|\alpha_{-1,p-s}\rangle\,.
  \end{align}
  By proposition \ref{sec:Vfusion} there exits a \(V_L\)-module isomorphism
  \begin{align}
    \mathcal{V}_{[1,2]}\fuse_{V_L}\mathcal{V}_{[1,p-s]}&\rightarrow \mathcal{V}_{[1,p-(s-1)]}\\\nonumber
    |\alpha_{1,2}\rangle\otimes|\alpha_{-1,p-s}\rangle&\mapsto |\alpha_{-1,p-(s-1)}\rangle\,.
  \end{align}
  By composing these two maps we have constructed a non-trivial \(\mathcal{W}_p\)-module map
  \begin{align}\label{eq:X2tosm1}
    \phi:X_2^+\fuse \mathcal{V}_{[1,p-s]}\rightarrow \mathcal{V}_{[1,p-(s-1)]}\,.
  \end{align}
  Also according to proposition \ref{sec:WtoV} there exists a surjective
  \(\mathcal{W}_p\)-module map
  \begin{align}
     \pi:\mathcal{V}_{[1,p-(s-1)]}&\rightarrow X_{s-1}^-\\\nonumber
    |\alpha_{-1,p-(s-1)}\rangle&\mapsto v_{s-1}^+\,.
  \end{align}
  The composition \(\pi\circ \phi\) is therefore a non-trivial \(\mathcal{W}_p\)-module map
  \begin{align}\label{eq:auxnontriv}
    \pi\circ\phi:X_2^+\fuse \mathcal{V}_{[1,p-s]}\rightarrow X_{s-1}^-\,.
  \end{align}
  By the surjection \eqref{eq:auxupperbound} \(X_2\fuse X_{p-s}^+\) must lie in
  the kernel of \(\pi\circ\phi\). Therefore there exists a non-trivial 
  \(\mathcal{W}_p\)-module map
  \begin{align}
    X_2^+\fuse X_s^-\rightarrow X_{s-1}^-\,.
  \end{align}
\end{proof}

\begin{proof}[Proof of step 4]
  According to proposition \ref{sec:WtoV} there exists an injective \(\mathcal{W}_p\)-module
  map
  \begin{align}
    X_2^+&\rightarrow \mathcal{V}_{[1,2]}\\\nonumber
    u_2&\mapsto|\alpha_{1,2}\rangle
  \end{align}
  and a \(\mathcal{W}_p\)-module isomorphism \(X_p^-\rightarrow \mathcal{V}_{[2,p]}\).
  By the above maps and proposition \ref{sec:subfusion}
  there exits a non-trivial \(\mathcal{W}_p\)-module map
  \begin{align}
    X_2^+\fuse \mathcal{V}_{[2,p]}&\rightarrow \mathcal{V}_{[1,2]}\fuse_{V_L}\mathcal{V}_{[2,p]}\\\nonumber
    u_2\otimes |\alpha_{0,p}\rangle&\mapsto |\alpha_{1,2}\rangle\otimes|\alpha_{0,p}\rangle\,.
  \end{align}
  By proposition \ref{sec:Vfusion} there exits a \(V_L\)-module isomorphism
  \begin{align}
    \mathcal{V}_{[1,2]}\fuse_{V_L}\mathcal{V}_{[2,p]}&\rightarrow \mathcal{V}_{[2,p+1]}=\mathcal{V}_{[1,1]}\\\nonumber
    |\alpha_{1,2}\rangle\otimes|\alpha_{0,p}\rangle&\mapsto |\alpha_{0,p+1}\rangle=|\alpha_{-1,1}\rangle\,.
  \end{align}
  By composing these two maps we have constructed a non-trivial \(\mathcal{W}_p\)-module map
  \begin{align}
    X_2^+\fuse X_p^-\rightarrow \mathcal{V}_{[1,1]}\,.
  \end{align}
  Also according to proposition \ref{sec:WtoV} there exists a surjective \(\mathcal{W}_p\)-module
  map
  \begin{align}
    \mathcal{V}_{[1,1]}&\rightarrow X_{p-1}^-\\\nonumber
    |\alpha_{-1,1}\rangle&\mapsto v_{p-1}^+\,.
  \end{align}
  Therefore there exists a non-trivial \(\mathcal{W}_p\)-module map
  \begin{align}
    X_2^+\fuse X_p^-\rightarrow X_{p-1}^-\,.
  \end{align}
\end{proof}

\begin{proof}[Proof of step 5]
  In order to prove the rigidity of \(X_2^+\), we need to consider
  three fold products of \(X_2^+\), which at this stage we can only
  compute for \(p\geq 4\). We will explain how the
  proof of rigidity can be reduced to analysing formal solutions of
  hypergeometric equations for \(p\geq 4\). The advantage of this analysis is
  that it
  does not require us to explicitly know the three fold fusion product
  of \(X_2^+\) and we can therefore also apply it to \(p=2,3\) once we have discussed \(p\geq 4\).

  Until explicitly stated otherwise we will therefore assume that \(p\geq 4\).
  Then we have proven that
  \begin{align}
    X_2^+\fuse (X_2^+\fuse X_2^+)\cong (X_2^+\fuse X_2^+)\fuse X_2^+=2\cdot X_2^+\oplus X_4^+\,.
  \end{align}
  The rigidity of \(X_2^+\) and the self-duality \({X_2^+}^\vee=X_2^+\) requires the existence of \(\mathcal{W}_p\)-module maps
  \(i:X_1^+\rightarrow X_2^+\fuse X_2^+\) and \(e:X_2^+\fuse X_2^+\rightarrow X_1^+\), such that
  \begin{align}\label{eq:commuA}
      \begin{tikzpicture}[baseline=(base)]
        \path
        (0,0) node (1) {\(X_2^+\cong X_2^+\fuse X_1^+ \)}
        ++(0,-2) node (2) {\(X_2^+\cong X_1^+\fuse X_2^+ \)}
        ++(5,2) node (3) {\(X_2^+\fuse(X_2^+\fuse X_2^+)\)}
        ++(0,-2) node (4) {\((X_2^+\fuse X_2^+)\fuse X_2^+\)};
        \node (0,-1.5) (base) {};
        \draw[->,>=latex] (1) -- (2) node[left,midway] {\(f\)};
        \draw[<-,>=latex] (2) -- (4) node[above,midway] {\(e\fuse\id\)};
        \draw[->,>=latex] (1) -- (3) node[above,midway] {\(\id\fuse i\)};
        \draw[->,>=latex] (3) -- (4) node[left,midway] {\(\alpha_{X_2^+,X_2^+,X_2^+}\)};
      \end{tikzpicture}
  \end{align}
  and
  \begin{align}\label{eq:commuB}
    \begin{tikzpicture}[baseline=(base)]
      \path
      (0,0) node (1) {\(X_2^+\cong X_1^+\fuse X_2^+ \)}
      ++(0,-2) node (2) {\(X_2^+\cong X_2^+\fuse X_1^+ \)}
      ++(5,2) node (3) {\((X_2^+\fuse X_2^+)\fuse X_2^+\)}
      ++(0,-2) node (4) {\(X_2^+\fuse (X_2^+\fuse X_2^+)\)};
      \node (0,-1.5) (base) {};
      \draw[->,>=latex] (1) -- (2) node[left,midway] {\(g\)};
      \draw[<-,>=latex] (2) -- (4) node[above,midway] {\(\id\fuse e\)};
      \draw[->,>=latex] (1) -- (3) node[above,midway] {\(i\fuse\id\)};
      \draw[->,>=latex] (3) -- (4) node[left,midway] {\(\alpha_{X_2^+,X_2^+,X_2^+}^{-1}\)};
    \end{tikzpicture}
  \end{align}
  commute, where \(f=\mu\cdot\id_{X_2^+},\ g=\nu\cdot\id_{X_2^+}\) for two non zero constants \(\mu\) and \(\nu\). 
  We show that \(\mu\neq0\), the case of \(\nu\) is similar so we omit the proof.

  We fix highest weight vectors \(u_s\) of \(X_s^+\) for \(s=1,2,3\), such that \(X_s^+[h_s^+]=\mathbb{C}u_s\) and \(u_1=\Omega\) 
  as in previous calculations.
  By the fusion products we have computed so far we know that the spaces of
  vertex operators
  \begin{align}
    \binom{X_2^+}{X_2^+,X_s^+},\quad\binom{X_2^+}{X_s^+,X_2^+},\quad \binom{X_s^+}{X_2^+,X_2^+}\,,
  \end{align}
  are all one dimensional for \(s=1,3\). We therefore fix non-trivial vertex operators
  \begin{align}
    \tensor*[_2]{\Psi}{^{2}_s}\in\binom{X_2^+}{X_2^+,X_s^+},\quad
    \tensor*[_s]{\Psi}{^{2}_2}\in\binom{X_2^+}{X_s^+,X_2^+},\quad 
    \tensor*[_2]{\Psi}{^{s}_2}\in\binom{X_s^+}{X_2^+,X_2^+}\,.
  \end{align}
  These vertex operators can be formally expanded as
  \begin{align}
    \tensor*[_b]{\Psi}{^{a}_c}(\ ; z)=
    \sum_{n\in\mathbb{Z}} \tensor*[_b]{\Psi}{^{a}_{c;n}}(\ )z^{-n-(h_a^++h_b^+-h_c^+)}\,,
  \end{align}
  for appropriate choices of \(a,b\) and \(c\), where
  \begin{align}
    \tensor*[_b]{\Psi}{^{a}_{c;n}}\in\bigoplus_{k,\ell\geq 0}\hom_{\mathbb{C}}(X_a^+[h_a^++k]\otimes
    X_b^+[h_b^++\ell],X_c^+[h_c^++k+\ell-n])\,.
  \end{align}
  This allows us to define four power series
  \begin{align}
    \Phi^{(1)}_s(z_4,z_3,z_2,z_1)&=\langle \Omega|\tensor*[_1]{\Psi}{^{2}_2}(u_2;z_4)
    \tensor*[_2]{\Psi}{^{2}_s}(u_2;z_3)\tensor*[_s]{\Psi}{^{2}_2}(u_2;z_2)
    \tensor*[_2]{\Psi}{^{2}_1}(u_2;z_1)\Omega\rangle\\\nonumber
    \Phi^{(2)}_s(z_4,z_3,z_2,z_1)&=\langle \Omega|\tensor*[_1]{\Psi}{^{2}_2}(u_2;z_4)
    \tensor*[_2]{\Psi}{^{s}_2}(\tensor*[_s]{\Psi}{^{2}_2}(u_2;z_2-z_3)u_2;z_3)
    \tensor*[_2]{\Psi}{^{2}_1}(u_2;z_1)\Omega\rangle\,,
  \end{align}
  for \(s=1,3\). The power series \(\Phi^{(1)}_s\) and \(\Phi^{(2)}_s\) converge absolutely on
  the domains
  \begin{align}
    U^{(1)}&=\{(z_4,z_3,z_2,z_1)\in (\mathbb{C}^\ast)^4\,|\,|z_4|>|z_3|>|z_2|>|z_1|>0\}\,,\\\nonumber
    U^{(2)}&=\{(z_4,z_3,z_2,z_1)\in (\mathbb{C}^\ast)^4\,|\,|z_4|>|z_3|>|z_1|>0,|z_3|>|z_2-z_3|>0\}\,,
  \end{align}
  respectively and satisfy the partial differential equations:
  \begin{enumerate}
  \item For \(n=-1,0,1\)
    \begin{align}\label{eq:moebiuscov}
      \sum_{a=1}^4 z_a^n\left(z_a\frac{\partial}{\partial z_a}+(n+1)h_2^+\right)\Phi =0\,.
    \end{align}
  \item For \(a=1,2,3,4\)
    \begin{align}\label{eq:nullrell}
      \left(\frac{\partial^2}{\partial z_a^2}-\frac{1}{p}\sum_{\substack{b=1\\b\neq a}}^4\left(
          \frac{h_2^+}{(z_b-z_a)^2}-\frac{1}{z_b-z_a}\frac{\partial}{\partial z_b}\right)\right)\Phi=0\,.
    \end{align}
  \end{enumerate}
  The first set of differential equations follows from M\"obius covariance and
  the second set from the null vector at level 2 in \(X_2^+\).

  The solution space of these two sets of differential equations is two dimensional and
  the solutions define multivalued holomorphic functions on \((\mathbb{P})^4\setminus\mathrm{diagonals}\). Therefore 
  \(\Phi_s^{(1)}\) and \(\Phi_s^{(2)}\) define bases of the solution space of the above
  differential equations on the two domains \(U^{(1)}\) and \(U^{(2)}\) and it is possible to analytically
  continue \(\Phi_s^{(1)}\) to \(U^{(2)}\) and vice versa. For a given path \(\gamma\) from 
  \(U^{(1)}\) to \(U^{(2)}\), \(\Phi_s^{(1)}\) can be written as a linear combination of 
  \(\Phi^{(2)}_1\) and \(\Phi_3^{(2)}\). This defines a connection matrix
  \begin{align}
    \left(
      \begin{array}{c}
        \Phi_{1}^{(1)}\\
        \Phi_{3}^{(1)}
      \end{array}\right)=\left(
      \begin{array}{cc}
        a&b\\
        c&d
      \end{array}\right)
    \left(
      \begin{array}{c}
        \Phi_{1}^{(2)}\\\Phi_{3}^{(2)}
      \end{array}\right)\,.
  \end{align}
  Going along the path \(\gamma\) in the opposite direction one can express \(\Phi_s^{2}\) as 
  a linear combination of \(\Phi_1^{(1)}\) and \(\Phi_3^{(1)}\) with the inverse of the connection
  matrix above
  \begin{align}
    \left(
      \begin{array}{c}
        \Phi_{1}^{(2)}\\
        \Phi_{3}^{(2)}
      \end{array}\right)=\left(
      \begin{array}{cc}
        a&b\\
        c&d
      \end{array}\right)^{-1}
    \left(
      \begin{array}{c}
        \Phi_{1}^{(1)}\\\Phi_{3}^{(1)}
      \end{array}\right)\,.
  \end{align}
  The constant \(\mu\), in \(f=\mu\cdot\id_{X_2^+}\) of diagram \eqref{eq:commuA}, being non-zero
  is equivalent to \(\Phi_1^{(1)}\) having non-vanishing contributions from \(\Phi_1^{(2)}\), {\it i.e.}
  \(a\) being non-zero. Similarly \(\nu\) is non-vanishing if \(\Phi_1^{(2)}\) has non-trivial contributions
  from \(\Phi_1^{(1)}\), which is the case when \(d\) is non-zero.

  The first set of differential equations \eqref{eq:moebiuscov} guarantees the covariance of
  \(\Phi_s^{(a)}\) with respect to M\"obius transformations. Since M\"obius transformations act 
  transitively on triples of pair wise distinct elements of \(\mathbb{P}\), we can fix three of the
  arguments of \(\Phi_s^{(a)}\) to uniquely determine 
  \begin{align}\label{eq:simpe4point}
    \Phi_s^{(a)}=\prod_{1\leq i< j\leq 4}(z_i-z_j)^{\frac{2p-3}{6p}}x^\frac{1}{3}(1-x)^\frac{1}{3}H_s^{(a)}(x)
  \end{align}
  up to a function \(H_s^{(a)}(x)\) of the 
  M\"obius invariant cross ratio
  \begin{align}
    x=\frac{z_4-z_3}{z_4-z_2}\frac{z_1-z_2}{z_1-z_3}\,.
  \end{align}
  See \cite{DiFrancesco:1997nk} for a wealth of examples regarding such computations.
  The functions \(H^{(1)}_s(x)\) and \(H^{(2)}_s(x)\) are absolutely convergent on \(1> |x| > 0\) and \(1>|1-x|>0\)
  respectively.
  The second set of differential equations \eqref{eq:nullrell} arise from the fact that the vertex operators
  above vanish upon inserting the null vector
  \begin{align}\label{eq:nullvec}
    (L_{-1}^2-\tfrac{1}{p}L_{-2})u_2\,.
  \end{align}
  The prefactors of \(H_s^{(a)}(x)\) in equation \eqref{eq:simpe4point} have been chosen such that the differential
  equation for \(H_s^{(a)}(x)\) induce by \eqref{eq:moebiuscov} is particularly simple. Namely the well known hypergeometric equations
  \begin{align}\label{eq:hypgeomeq}
    x(1-x)\frac{\d^2}{\d x^2}H^{(a)}_s(x)+\tfrac{2}{p}(1-2x)\frac{\d}{\d x}H^{(a)}_s(x)-\tfrac{3-p}{p^2}H_s^{(a)}(x)=0\,.
  \end{align}
  For a detailed list of solutions and all formulae we will be using see \cite{Abramowitz:1964}.
  For \(\Phi_s^{(1)}\) which converges on \(U^{(1)}\), \(H_s^{(1)}(x)\) is a power series in \(x\),
  while for \(\Phi^{(2)}_s\) which converges on \(U^{(2)}\), \(H_s^{(2)}(x)\) is a power series in \(1-x\)
  \begin{align}
    H_1^{(1)}(x)&=\tensor[_2]{F}{_1}(\tfrac1p,\tfrac{3-p}{p};\tfrac2p;x)\,,\\\nonumber
    H_3^{(1)}(x)&=x^{\tfrac{p-2}{p}}\tensor[_2]{F}{_1}(\tfrac{p-1}{p},\tfrac1p;\tfrac{2p-2}{p};x)\,,\\\nonumber
    H_1^{(2)}(x)&=\tensor[_2]{F}{_1}(\tfrac1p,\tfrac{3-p}{p};\tfrac2p;1-x)\,,\\\nonumber
    H_3^{(2)}(x)&=(1-x)^{\tfrac{p-2}{p}}\tensor[_2]{F}{_1}(\tfrac1p,\tfrac{p-1}{p};\tfrac{2p-2}{p};1-x)\,.
  \end{align}

  To prove that \(\Phi_1^{(1)}\) has non-vanishing contributions from \(\Phi_1^{(2)}\) we continue
  \(H_1^{(1)}(x)\) along the path from 0 to 1 on the real line. The well known connection formula
  for hypergeometric functions then yields
  \begin{align}
    H_1^{(1)}(x)&=\frac{1}{2\cos \tfrac{\pi}{p}}H_1^{(2)}(x)
    +\frac{3-p}{2-p}\frac{\Gamma(\tfrac{2}{p})^2}{\Gamma(\tfrac{1}{p})\Gamma(\tfrac{3}{p})}H_3^{(2)}(x)\,,\\\nonumber
    H_1^{(2)}(x)&=\frac{1}{2\cos \tfrac{\pi}{p}}H_1^{(1)}(x)
    +\frac{3-p}{2-p}\frac{\Gamma(\tfrac{2}{p})^2}{\Gamma(\tfrac{1}{p})\Gamma(\tfrac{3}{p})}H_3^{(1)}(x)\,.
  \end{align}
  And thus the rigidity of \(X_2^+\) for \(p\geq 4\) follows.

  For \(p=2,3\) the analysis is exactly the same. Specifying the domains and codomains of the vertex
  operators is just a bit trickier. The resulting differential equations are analogous however.
  For \(p=3\) we have shown so far that
  \begin{align}
    X_2^+\fuse X_2^+\fuse X_2^+= X_2^+\oplus (X_2^+\fuse X_3^+)
  \end{align}
  and that there exits a surjective \(\mathcal{W}_3\)-module map
  \begin{align}
    X_2^+\fuse X_3^+\rightarrow X_2^+\,.
  \end{align}
  Therefore the right exactness of the fusion product implies the existence of a surjective \(\mathcal{W}_3\)-module
  map
  \begin{align}
    X_2^+\fuse (X_2^+\fuse X_3^+)\rightarrow X_1^+\oplus X_3^+\,.
  \end{align}
  The analysis above can therefore be repeated for \(p=3\) without any modifications.
  We consider the differential equations \eqref{eq:moebiuscov} and
  \eqref{eq:nullrell}, Which can again be simplified to the hypergeometric equation \eqref{eq:hypgeomeq}. Analysing
  the connection formulae for \(p=3\) yields
  \begin{align}
    H_1^{(1)}(x)=H_1^{(2)}(x)\,,
  \end{align}
  thus proving the rigidity of \(X_2^+\) for \(p=3\).

  For \(p=2\) the space of solutions for the hypergeometric equation
  \begin{align}
    x(1-x)\frac{\d^2}{\d x^2}H^{(a)}_s(x)+(1-2x)\frac{\d}{\d x}H^{(a)}_s(x)-\tfrac{1}{2}H^{(a)}_s(x)=0
  \end{align}
  is slightly more complicated than in the previous examples, because the poles encountered at \(x=0\) and \(x=1\) are
  logarithmic. This implies that vertex operators involved also contain logarithms. We will omit the details however
  since they are not important for solving the above differential equation. The solutions \(H_s^{(a)}\) are given by
    \begin{align}
    H_1^{(1)}(x)&=\tensor[_2]{F}{_1}(\tfrac12,\tfrac{1}{2};1;x)\,,\\\nonumber
    H_3^{(1)}(x)&=\tensor[_2]{F}{_1}(\tfrac12,\tfrac{1}{2};1;x)\log(x)+G(x)\,,\\\nonumber
    H_1^{(2)}(x)&=\tensor[_2]{F}{_1}(\tfrac12,\tfrac{1}{2};1;1-x)\,,\\\nonumber
    H_3^{(2)}(x)&=\tensor[_2]{F}{_1}(\tfrac12,\tfrac{1}{2};1;1-x)\log(1-x)+G(1-x)\,,
  \end{align}
  where \(G(x)\) is a power series with vanishing constant term that converges for \(1>|x|\). The connection
  formulae for \(p=2\) yield
  \begin{align}
    H_1^{(1)}=\frac{\log(4)}{\pi} H^{(2)}_1(x)-\frac{1}{\pi} H_3^{(2)}(x)\,,
  \end{align}
  thus proving the rigidity of \(X_2^+\) for \(p=2\).
\end{proof}

\section{\boldmath The rigidity of  \((\mathcal{W}_p\text{-mod},\fuse)\)\unboldmath}
\label{sec:wpisrigid}

In the previous sections we proved that \(X_1^-\) and \(X_2^+\) are rigid self dual objects in \(\mathcal{W}_p\)-mod.
In this section we will exploit this fact to compute the fusion
product of \(X_1^-\) and \(X_2^+\) with the projective modules \(P_s^\epsilon,\ 1\leq s<p,\ \epsilon=\pm\), allowing
us to prove the rigidity of \(\mathcal{W}_p\)-mod and ultimately compute the fusion product on the set of all
simple and all projective modules.

\subsection{\boldmath Fusion products between \(X_1^-\) and \(X_2^+\) and projective modules\unboldmath}

At first we prepare some more notation.
For any object \(Z\) in \(\mathcal{W}_p\)-mod we denote by \([Z:X_s^\epsilon]\) the multiplicity of
\(X_s^\epsilon\) in quotients \(M_{i+1}(Z)/M_i(Z)\) of the Jordan-H\"older
series \eqref{eq:JordanHoelder} of \(Z\), that is,
\begin{align}
  [Z:X_s^\epsilon]=\dim_\mathbb{C}\hom(P_s^\epsilon,Z)\,.
\end{align}

We have established that
\begin{align}
  X_2^+\fuse X_s^\epsilon&=\left\{
    \begin{array}{ll}
      X_2^\epsilon\,,&s=1\\
      X_{s-1}^\epsilon\oplus X_{s+1}^\epsilon\,,&2\leq s\leq p-1\\
          P_{p-1}^\epsilon\,,&s=p
    \end{array}\right.\\\nonumber
  X_1^-\fuse X_s^\epsilon&=X_s^{-\epsilon}\,,\qquad 1\leq s\leq p\,.
\end{align}
and that \(X_1^-\) and \(X_2^+\) are self-dual rigid objects. From the Jordan-H\"older series of the projective modules
we also know that
\begin{align}
  [P_s^\epsilon:X_t^\sigma]&=2\delta_{(s,\epsilon),(t,\sigma)}
  +2\delta_{(s,\epsilon),(p-t,-\sigma)}\,,\quad 1\leq s <p\\\nonumber
  [P_p^\epsilon:X_t^\sigma]&=\delta_{(p,\epsilon),(t,\sigma)}\,.
\end{align}
\begin{prop}
  The fusion rules of \(X_2^+\) and \(X_1^-\) with projective modules are
  given by
  \begin{align}\label{eq:projfusion}
    X_2^+\fuse P_s^\mu&=\left\{
      \begin{array}{ll}
        P^\mu_2\oplus 2\cdot P_p^{-\mu}\,,&s=1\\
        P^\mu_{s-1}\oplus P^\mu_{s+1}\,,&1<s<p-1\\
        P^\mu_{p-2}\oplus 2\cdot P_p^\mu\,,&s=p-1
      \end{array}\right.\\\nonumber
    X_1^-\fuse P_s^\mu&=P_s^{-\mu}\,,\qquad 1\leq s\leq p\,.
  \end{align}
\end{prop}
\begin{proof}
  Because \(X_1^-\) and \(X_2^+\) are rigid, their product with \(P_t^\delta\) is 
  projective. The most general ansatz for such a product is therefore
  \begin{align}
    X\fuse P_t^\delta=\bigoplus_{m=1}^p\bigoplus_{\mu=\pm} N_{m,\mu}\cdot P_m^\mu\,,
  \end{align}
  where \(X\) is either \(X_1^-\) or \(X_2^+\) and \(N_{m,\mu}\in\mathbb{Z}\) is the multiplicity of \(P_m^\mu\) in
  \(X\fuse P_t^\delta\). We can determine
  \(N_{m,\mu}\) by recalling that a rigid object \(X\) and two arbitrary objects \(A\) and \(B\)
  satisfy the relation
  \begin{align}
    \hom(A,X\fuse B)\cong \hom(X^\ast\fuse A,B)\,.
  \end{align}
  Setting \(A\) to \(P_t^\delta\) and \(C\) to \(X_m^\mu\) and calculating the dimensions of the spaces
  of \(\mathcal{W}_p\)-module maps in the equation above, we are lead to
  \begin{align}
    N_{m,\mu}=\dim\hom(P_t^\sigma,X\fuse X_m^\mu)=[X\fuse X_m^\mu:X_t^\sigma]\,.
  \end{align}

  We can easily calculate the multiplicities \([X\fuse X_m^\mu:X_t^\sigma]\) for \(X=X_1^-,\ X_2^+\)
  by considering the fusion products
  \ref{sec:X1fusion} and \ref{sec:X2Xsthm}
  \begin{align}
    [X_1^-\fuse X_m^\mu:X_t^\delta]&=\delta_{(t,\delta),(m,-\mu)}\\\nonumber
    [X_2^+\fuse X_m^\mu:X_t^\delta]&=\left\{
      \begin{array}{ll}
        \delta_{(2,+),(m,\mu)}+2\delta_{(p,\delta),(m,-\mu)}&t=1\\
        \delta_{(t-1,\delta),(m,\mu)}+\delta_{(t+1,\delta),(m,\mu)}&2\leq t\leq p-2\\
        \delta_{(p-2,\delta),(m,\mu)}+2\delta_{(p,\delta),(m,\mu)}&t=p-1\\
        \delta_{(p-1,\delta),(m,\mu)}&t=p
      \end{array}
\right.\,.
  \end{align}
  The proposition then follows directly by plugging in the multiplicities.
\end{proof}

\subsection{Proving rigidity}

We apply point 4 of proposition \ref{sec:moncatprop} to \(\mathcal{W}_p\)-mod. Since
all simple and all projective \(\mathcal{W}_p\)-modules appear in the repeated
fusion products  of \(X_1^-\)
and \(X_2^+\) we have the following proposition.
\begin{prop}
  For \(1\leq s\leq p,\ \epsilon=\pm\) the simple modules \(X_s^\epsilon\) and the projective
  modules \(P_s^\epsilon\) are self-dual rigid objects in \(\mathcal{W}_p\)-mod.
\end{prop}
In \(\mathcal{W}_p\)-mod all indecomposable objects \(M\) except the simple objects and
the projective objects satisfy exact sequences
\begin{align}
  0\longrightarrow L\longrightarrow M\longrightarrow N\longrightarrow 0
\end{align}
such that \(L\) and \(N\) are direct sums of simple objects.
So finally we obtain the rigidity of \(\mathcal{W}_p\)-mod by applying point 5 of proposition \ref{sec:moncatprop}.
\begin{thm}\label{sec:wprigid}
  The weakly rigid monoidal category \((\mathcal{W}_p,\fuse,X_1^+)\) is rigid.
  For any object \(M\) in \(\mathcal{W}_p\)-mod the dual \(M^\vee\) is given
  by the contragredient \(M^\ast\), {\it i.e.} \(M^\vee=M^\ast\).
\end{thm}

\subsection{\boldmath The ring structure of \(P(\mathcal{W}_p)\) and \(K(\mathcal{W}_p)\)\unboldmath}

We see in theorems \ref{sec:X1fusion} and \ref{sec:X2Xsthm}, that the fusion products of
\(X_1^-\) and \(X_2^+\) with simple modules are direct sums of simple and projective modules. Because all
simple modules appear as direct summands of products of \(X_1^-\) and \(X_2^+\), the
product of any two simple modules must also be a product of simple and projective modules.
Therefore because all the simple modules are rigid, the fusion product closes on the set
of all simple and all projective modules. In this section we will compute the fusion product
on this set.

We introduce the free abelian group \(P(\mathcal{W}_p)\) of rank \(4p-2\) generated by all projective
and all simple modules
\begin{align}
  P(\mathcal{W}_p)=\bigoplus_{s=1}^p\bigoplus_{\epsilon=\pm} \mathbb{Z}[X_s^\epsilon]_P
  \oplus\bigoplus_{s=1}^{p-1}\bigoplus_{\epsilon=\pm}\mathbb{Z}[P_s^\epsilon]_P
\end{align}
and the rank \(2p\) Grothendieck group\footnote{The Grothendieck group \(K(\mathcal{C})\) can be defined for any abelian category \(\mathcal{C}\).
It is given by free abelian group generated by all objects of \(\mathcal{C}\) modulo the subgroup generated by all formal differences
\(M-L-N\) where \(L,M,N\) satisfy an exact sequence \(0\longrightarrow L\longrightarrow M\longrightarrow N\longrightarrow 0\). If the number
of simple objects in the abelian category \(\mathcal{C}\) is finite, then \(K(\mathcal{C})\) is just the finite rank free abelian group
generated by all simple objects.}
\begin{align}
  K(\mathcal{W}_p)=\bigoplus_{s=1}^p\bigoplus_{\epsilon=\pm}\mathbb{Z}[X_s^\epsilon]_K\,.
\end{align}
By the rigidity of \(\mathcal{W}_p\)-mod and and the closure of the fusion product
on simple and projective modules:
\begin{enumerate}
\item \(P(\mathcal{W}_p)\) and \(K(\mathcal{W}_p)\) have the structure of commutative rings.
\item The canonical projection \(\pi:P(\mathcal{W}_p)\rightarrow K(\mathcal{W}_p)\) is a ring homomorphism.
\end{enumerate}

By the above arguments the two operators
\begin{align}
  X&=X_2^+\fuse-\,,&Y&=X_1^-\fuse-\,,
\end{align}
define \(\mathbb{Z}\)-linear endomorphisms of \(P(\mathcal{W}_p)\) and \(K(\mathcal{W}_p)\). Because the fusion product
is commutative, the two operators \(X\) and \(Y\) must also commute. Thus
by the two operators \(X\) and \(Y\) the polynomial ring \(\mathbb{Z}[X,Y]\) acts on \(P(\mathcal{W}_p)\) and \(K(\mathcal{W}_p)\),
{\it i.e.} \(P(\mathcal{W}_p)\) and \(K(\mathcal{W}_p)\) are modules over \(\mathbb{Z}[X,Y]\) and the canonical projection
\(\pi\) is a \(\mathbb{Z}[X,Y]\)-module map.

Before we begin analysing the action of \(\mathbb{Z}[X,Y]\) on \(P(\mathcal{W}_p)\) we recall
some elementary facts about Chebyshev polynomials that will prove helpful.
\begin{definition}
  We define elements \(U_n(X),\ n=0,1,\dots\) in \(\mathbb{Z}[X]\) recursively
  \begin{align}\label{eq:chebypol}
    U_0(X)&=1\,,\qquad U_1(X)=X\,,\\\nonumber
    U_{n+1}(X)&=XU_n(X)-U_{n-1}(X)\,.
  \end{align}
\end{definition}
\begin{remark}
  \begin{enumerate}
  \item The coefficient of the leading order of \(U_n(X)\) is 1, {\it i.e.}
    \begin{align}
      U_n(X)=X^n+\cdots\in\mathbb{Z}[X],\ m=0,1,2,\dots
    \end{align}
    so we have
    \begin{align}
      \mathbb{Z}[X]=\bigoplus_{n=0}^\infty \mathbb{Z}U_n(X)\,.
    \end{align}
  \item The initial conditions and recursion relations of the polynomials \(U_n(X)\) are those
    of the Chebyshev polynomials of the second kind, though with a non-standard choice of normalisation.
  \end{enumerate}
\end{remark}

We define the \(\mathbb{Z}[X,Y]\)-module maps
\begin{align}
  \psi:\mathbb{Z}[X,Y]&\rightarrow P(\mathcal{W}_p)\,,\\\nonumber
  f(X,Y)&\mapsto f(X,Y)\cdot [X_1^+]_P\\\nonumber
  \phi:\mathbb{Z}[X,Y]&\rightarrow K(\mathcal{W}_p)\,.\\\nonumber
  f(X,Y)&\mapsto f(X,Y)\cdot [X_1^+]_K
\end{align}
\begin{thm}
  The maps \(\psi\) and \(\phi\) are surjective homomorphisms of
  commutative rings and the kernels are given by the ideals
  \begin{align}
    \ker \psi&=\left<Y^2-1,U_{2p-1}(X)-2YU_{p-1}(X)\right>\,,\\\nonumber
    \ker \phi&=\left<Y^2-1,U_{p}(X)-U_{p-2}(X)-2Y\right>\,.
  \end{align}
\end{thm}

\begin{proof}
  Consider the fusion products
  \begin{align}
    X_2^+\fuse X_1^+&=X_2^+\,,\\\nonumber
    X_2^+\fuse X_s^+&=X_{s-1}^+\oplus X_{s+1}^+\,,
  \end{align}
  for \(1<s<p\).
  Formally this looks exactly like the recursion relations and initial
  conditions \eqref{eq:chebypol} if one were to substitute \(X_2^+\) with \(X\) and \(X_s^+\) with \(U_{s-1}(X)\). We can
  therefore write the generators of \(P(\mathcal{W}_p)\) and \(K(\mathcal{W}_p)\) corresponding to the simple modules
  \(X_s^+,\ 1\leq s\leq p\) as
  \begin{align}
    [X_s^+]_P&=U_{s-1}(X)[X_1^+]_P\,,&[X_s^+]_K&=U_{s-1}(X)[X_1^+]_K\,.
  \end{align}
  Since the remaining simple modules \(X_s^-\) can be written as \(X_1^-\fuse X_s^+\) their corresponding
  generators in \(P(\mathcal{W}_p)\) and \(K(\mathcal{W}_p)\) can be written as
  \begin{align}
    [X_s^-]_P&=YU_{s-1}(X)[X_1^+]_P\,,&[X_s^-]_K&=YU_{s-1}(X)[X_1^+]_K\,.
  \end{align}
  Thus as a module over \(\mathbb{Z}[X,Y]\), \(K(\mathcal{W}_p)\) is generated by \([X_1^+]_K\)
  and \(\phi\) is therefore a surjective \(\mathbb{Z}[X,Y]\)-homomorphism.

  Next we consider the fusion products
  \begin{align}
    X_2\fuse X_p^+&=P_{p-1}^+\,,\\\nonumber
    X_2^+\fuse P_{s}^+&=P_{s-1}^+\oplus P_{s+1}^+\,,
  \end{align}
  for \(1<s<p\). These imply that the generators of \(P(\mathcal{W}_p)\) corresponding
  to the projective modules \(P_s^+,\ 1\leq s<p\) can be written as
  \begin{align}\label{eq:projident}
    [P_s^+]_P=(U_{2p-1-s}(X)+U_{s-1}(X))[X_1^+]_P\,.
  \end{align}
  Since the remaining projective modules \(P_s^-\) can be written as \(X_1^-\fuse P_s^+\)
  their corresponding generators in \(P(\mathcal{W}_p)\) can be written as
  \begin{align}
    [P_s^-]_P=Y(U_{2p-1-s}(X)+U_{s-1}(X))[X_1^+]_P\,.
  \end{align}
    Thus as a module over \(\mathbb{Z}[X,Y]\), \(P(\mathcal{W}_p)\) is generated by \([X_1^+]_P\)
  and \(\psi\) is therefore a surjective \(\mathbb{Z}[X,Y]\)-homomorphism.

  We verify that the two ideals
  \begin{align}
    I&=\left<Y^2-1,U_{2p-1}(X)-2YU_{p-1}(X)\right>\,,\\\nonumber
    J&=\left<Y^2-1,U_{p}(X)-U_{p-2}(X)-2Y\right>\,,
  \end{align}
  are indeed the kernels of \(\psi\) and \(\phi\) by showing that \(I\) and \(J\) lie in the kernels
  and that the ranks of \(\mathbb{Z}(X,Y)/I\) and \(\mathbb{Z}(X,Y)/J\) are equal to the ranks of \(P(\mathcal{W}_p)\)
  and \(K(\mathcal{W}_p)\).
  From the fusion product \(X_1^-\fuse X_1^-=X_1^+\) it follows that
  \begin{align}
    (Y^2-1)[X_1^+]_P&=0\,,&(Y^2-1)[X_1^+]_K&=0\,.
  \end{align}
  The left and right hand sides of \(X_2^+\fuse X_p^+=P_{p-1}^+\) are given by the
  left and right hand sides of
  \begin{align}
    XU_{p-1}(X)[X_1^+]_K=2(U_{p-1}(X)+Y)[X_1^+]_K\,,
  \end{align}
  respectively in \(K(\mathcal{W}_p)\).
  By the recursion relations for Chebyshev polynomials it therefore follows that
  \begin{align}
    (U_{p}(X)-U_{p-2}(X)-2Y)[X_1^+]_K=0\,.
  \end{align}
  Lastly by the left and right hand sides of the product \(X_2^+\fuse P_{1}^+=P_2^+\oplus 2X_p^-\)
  are given by the left and right hand sides of
  \begin{align}
    X(U_{2p-2}(X)+U_0(X))[X_1^+]_P=(U_{2p-3}(X)+U_1(X)+2YU_{p-1}(X))[X_1^+]_P
  \end{align}
  respectively in \(P(\mathcal{W}_p)\).
  By the recursion relations for Chebyshev polynomials it therefore follows that
  \begin{align}
    (U_{2p-1}(X)-2YU_{p-1}(X))[X_1^+]_P=0\,.
  \end{align}

  We decompose \(\mathbb{Z}[X,Y]/I\) and \(\mathbb{Z}[X,Y]/J\) as free abelian groups to compute their rank
  \begin{align}
    \frac{\mathbb{Z}[X,Y]}{I}&=\frac{\mathbb{Z}[X]\oplus \mathbb{Z}[X]Y}{\left<U_{2p-1}(X)-2YU_{p-1}(X)\right>}
    =\bigoplus_{i=0}^{2p-2}\mathbb{Z}X^i\oplus \bigoplus_{i=0}^{2p-2}\mathbb{Z}X^iY\\\nonumber
    \frac{\mathbb{Z}[X,Y]}{J}&=\frac{\mathbb{Z}[X]\oplus \mathbb{Z}[X]Y}{\left<U_{p}(X)-U_{p-2}(X)-2Y\right>}
    =\bigoplus_{i=0}^{p-1}\mathbb{Z}X^i\oplus \bigoplus_{i=0}^{p-1}\mathbb{Z}X^iY
  \end{align}
  and see that the ranks are \(4p-2\) and \(2p\) respectively.
\end{proof}

\begin{thm}\label{sec:fusionrules}
  The fusion products for all simple and all projective \(\mathcal{W}_p\)-modules are given by
  \begin{align}\label{eq:frules}
    X_s^\mu\fuse X_t^\nu&=\!\!\!\!\!\!\!\!\!\!\!\!\!\!\!\!\bigoplus_{i=|s-t|+1;2}^{\text{min}\{s+t-1,2p-1-s-t\}}
    \!\!\!\!\!\!\!\!\!\!\!\!\!\!\!\!\!X_i^{\mu\nu}
    \ \ \oplus\!\!\!\!\bigoplus_{i=2p+1-s-t;2}^{m}\!\!\!\!\!\!\!\!\!P_i^{\mu\nu}\\\nonumber
    X_s^\mu\fuse P_t^\nu&=\!\!\!\!\!\!\!\!\!\!\!\!\!\!\!\!\bigoplus_{i=|s-t|+1;2}^{\text{min}\{s+t-1,2p-1-s-t\}}
    \!\!\!\!\!\!\!\!\!\!\!\!\!\!\!\!\!P_i^{\mu\nu}
    \ \ \oplus\!\!\!\!\bigoplus_{i=2p+1-(s+t));2}^m\!\!\!\!\!\!\!\!\!\!\!2\cdot P_i^{\mu\nu}
    \ \ \oplus\!\!\!\!\bigoplus_{i=p+1-(s-t));2}^m\!\!\!\!\!\!\!\!\!\!\!2\cdot P_i^{-\mu\nu}\\\nonumber
    P_s^\mu\fuse P_t^\nu&=2\cdot X_s^+\fuse P_t^+\oplus 2\cdot X_{p-s}^-\fuse P_t^+
    \,,
  \end{align}
  where ``\(;2\)'' indicates that the summation variable is incremented in steps of 2 and
  \begin{align}
    m=\left\{
      \begin{array}{ll}
        p&\text{if }p-i \text{ is even}\\
        p-1&\text{if }p-i \text{ is odd}
      \end{array}
      \right.\,.
  \end{align}
  The product on the Grothendieck group induced by the fusion product is given by
  \begin{align}
  [X_s^\mu]_K\cdot[X_t^\nu]_K=\sum_{i=|s-t|+1;2}^{a(s+t)}
    [X_i^{\mu\nu}]_K
    &+\sum_{i=b(s+t);2}^{s+t-1-p}
    2[X_{i}^{-\mu\nu}]_K+[X_{p-i}^{\mu\nu}]_K\,,
  \end{align}
  where
  \begin{align}
    a(s)&=\left\{
      \begin{array}{ll}
        s-1&\text{if } s-1-p\leq 0\\
        p&\text{if } s-1-p>0 \text{ is odd}\\
        p-1&\text{if } s-1-p>0 \text{ is even}
      \end{array}
      \right.\,,
    \\
    b(s)&=\left\{
      \begin{array}{ll}
        1&\text{if } s-1-p \text{ is even}\\
        2&\text{if } s-1-p \text{ is odd}
      \end{array}
      \right.\,.
  \end{align}
\end{thm}
\begin{proof}
  The above fusion rules can be computed directly in \(\mathbb{Z}[X,Y]\) by using multiplication formula for
  Chebyshev polynomials
  \begin{align}
    U_k(x)U_j(x)=\sum_{i=|k-j|;2}^{k+j}U_i(x)\
  \end{align}
  and subsequently projecting onto \(P(\mathcal{W}_p)\) or \(K(\mathcal{W}_p)\).
  Note that the ``;2'' in the subscript of the sum indicates that the summation variable \(k\) is incremented in steps of 2.
\end{proof}

\bibliographystyle{utphys}
\bibliography{wp_mod_tensor_cat}

\providecommand{\href}[2]{#2}\begingroup\raggedright\begin{thebibliography}{10}

\bibitem{Borcherds:1983sq}
R.~E. Borcherds, ``{Vertex algebras, Kac-Moody algebras, and the monster},''
\href{http://dx.doi.org/10.1073/pnas.83.10.3068}{{\em Proc. Nat. Acad. Sci.}
  {\bfseries 83} (1986) 3068--3071}.

\bibitem{Frenkel:1988xz}
I.~Frenkel, J.~Lepowsky, and A.~Meurman, ``{Vertex Operator Algebras and the
  Monster},''
{\em Pure and Appl. Math, Academic Press, Boston} {\bfseries 134} (1988) .

\bibitem{Frenkel:1992}
I.~B. Frenkel and Y.~Zhu, ``{Vertex operator algebras associated to
  representations of affine and Virasoro algebras},''
  \href{http://dx.doi.org/10.1215/S0012-7094-92-06604-X}{{\em Duke Math J.}
  {\bfseries 66} (1992) 123--168}.

\bibitem{Frenkel:1993}
I.~Frenkel, Y.-Z. Huang, and J.~Lepowsky, ``{On Axiomatic Approaches to Vertex
  Operator Algebras and Modules},''
{\em Memoirs Amer. Math. Soc.} {\bfseries 104} (1993) .

\bibitem{Zhu:1996}
Y.~Zhu, ``{Modular invariance of characters of vertex operator algebras},''
  \href{http://dx.doi.org/10.1090/S0894-0347-96-00182-8}{{\em J. Amer. Math.
  Soc.} {\bfseries 9} (1996) 237--302}.

\bibitem{Frenkel:2001}
E.~Frenkel and D.~Ben-Zvi, ``{Vertex Algebras and Algebraic Curves},'' {\em
  Mathematical Surveys and Monographs, Amer. Math. Soc.} {\bfseries 88} (2001)
  .

\bibitem{Jeng:2006tg}
M.~Jeng, G.~Piroux, and P.~Ruelle, ``{Height variables in the Abelian sandpile
  model: scaling fields and correlations},'' {\em J. Stat. Mech.} {\bfseries
  0610} (2006) P015,
\href{http://arxiv.org/abs/cond-mat/0609284}{{\ttfamily
  arXiv:cond-mat/0609284}}.

\bibitem{Pearce:2006we}
P.~A. Pearce and J.~Rasmussen, ``{Solvable critical dense polymers},'' {\em J.
  Stat. Mech.} {\bfseries 0702} (2007) P015,
\href{http://arxiv.org/abs/hep-th/0610273}{{\ttfamily arXiv:hep-th/0610273}}.

\bibitem{Read:2007qq}
N.~Read and H.~Saleur, ``{Associative-algebraic approach to logarithmic
  conformal field theories},''
  \href{http://dx.doi.org/10.1016/j.nuclphysb.2007.03.033}{{\em Nucl. Phys.}
  {\bfseries B777} (2007) 316--351},
\href{http://arxiv.org/abs/hep-th/0701117}{{\ttfamily arXiv:hep-th/0701117}}.

\bibitem{Mathieu:2007pe}
P.~Mathieu and D.~Ridout, ``{From Percolation to Logarithmic Conformal Field
  Theory},'' \href{http://dx.doi.org/10.1016/j.physletb.2007.10.007}{{\em Phys.
  Lett.} {\bfseries B657} (2007) 120--129},
\href{http://arxiv.org/abs/0708.0802}{{\ttfamily arXiv:0708.0802 [hep-th]}}.

\bibitem{Ridout:2008cv}
D.~Ridout, ``{On the Percolation BCFT and the Crossing Probability of Watts},''
  \href{http://dx.doi.org/10.1016/j.nuclphysb.2008.09.038}{{\em Nucl. Phys.}
  {\bfseries B810} (2009) 503--526},
\href{http://arxiv.org/abs/0808.3530}{{\ttfamily arXiv:0808.3530 [hep-th]}}.

\bibitem{Nigro:2009si}
A.~Nigro, ``{Integrals of Motion for Critical Dense Polymers and Symplectic
  Fermions},'' \href{http://dx.doi.org/10.1088/1742-5468/2009/10/P10007}{{\em
  J. Stat. Mech.} {\bfseries 0910} (2009) P10007},
\href{http://arxiv.org/abs/0903.5051}{{\ttfamily arXiv:0903.5051 [hep-th]}}.

\bibitem{Gurarie:1993xq}
V.~Gurarie, ``{Logarithmic operators in conformal field theory},''
  \href{http://dx.doi.org/10.1016/0550-3213(93)90528-W}{{\em Nucl. Phys.}
  {\bfseries B410} (1993) 535--549},
\href{http://arxiv.org/abs/hep-th/9303160}{{\ttfamily arXiv:hep-th/9303160}}.

\bibitem{Gaberdiel:1996np}
M.~R. Gaberdiel and H.~G. Kausch, ``{A rational logarithmic conformal field
  theory},'' \href{http://dx.doi.org/10.1016/0370-2693(96)00949-5}{{\em Phys.
  Lett.} {\bfseries B386} (1996) 131--137},
\href{http://arxiv.org/abs/hep-th/9606050}{{\ttfamily arXiv:hep-th/9606050}}.

\bibitem{Fuchs:2003yu}
J.~Fuchs, S.~Hwang, A.~M. Semikhatov, and I.~Y. Tipunin, ``{Nonsemisimple
  fusion algebras and the Verlinde formula},''
  \href{http://dx.doi.org/10.1007/s00220-004-1058-y}{{\em Commun. Math. Phys.}
  {\bfseries 247} (2004) 713--742},
\href{http://arxiv.org/abs/hep-th/0306274}{{\ttfamily arXiv:hep-th/0306274}}.

\bibitem{Huang:2003za}
Y.-Z. Huang, J.~Lepowsky, and L.~Zhang, ``{A Logarithmic generalization of
  tensor product theory for modules for a vertex operator algebra},''
  \href{http://dx.doi.org/10.1142/S0129167X06003758}{{\em Int.J.Math.}
  {\bfseries 17} (2007) 975--1012},
\href{http://arxiv.org/abs/math/0311235}{{\ttfamily arXiv:math/0311235
  [math-qa]}}.

\bibitem{Carqueville:2005nu}
N.~Carqueville and M.~Flohr, ``{Nonmeromorphic operator product expansion and
  C(2)-cofiniteness for a family of W-algebras},''
  \href{http://dx.doi.org/10.1088/0305-4470/39/4/015}{{\em J.Phys. A}
  {\bfseries 39} (2006) 951},
\href{http://arxiv.org/abs/math-ph/0508015}{{\ttfamily arXiv:math-ph/0508015
  [math-ph]}}.

\bibitem{Feigin:2005zx}
B.~Feigin, A.~Gainutdinov, A.~Semikhatov, and I.~Tipunin, ``{Modular group
  representations and fusion in logarithmic conformal field theories and in the
  quantum group center},''
  \href{http://dx.doi.org/10.1007/s00220-006-1551-6}{{\em Commun. Math. Phys.}
  {\bfseries 265} (2006) 47--93},
\href{http://arxiv.org/abs/hep-th/0504093}{{\ttfamily arXiv:hep-th/0504093
  [hep-th]}}.

\bibitem{Feigin:2006iv}
B.~L. Feigin, A.~M. Gainutdinov, A.~M. Semikhatov, and I.~Y. Tipunin,
  ``{Logarithmic extensions of minimal models: Characters and modular
  transformations},''
  \href{http://dx.doi.org/10.1016/j.nuclphysb.2006.09.019}{{\em Nucl. Phys.}
  {\bfseries B757} (2006) 303--343},
\href{http://arxiv.org/abs/hep-th/0606196}{{\ttfamily arXiv:hep-th/0606196}}.

\bibitem{Gaberdiel:2007jv}
M.~R. Gaberdiel and I.~Runkel, ``{From boundary to bulk in logarithmic CFT},''
  \href{http://dx.doi.org/10.1088/1751-8113/41/7/075402}{{\em J. Phys. A}
  {\bfseries 41} (2008) 075402},
\href{http://arxiv.org/abs/0707.0388}{{\ttfamily arXiv:0707.0388 [hep-th]}}.

\bibitem{Adamovic:2007er}
D.~Adamovic and A.~Milas, ``{On the triplet vertex algebra W(p)},''
  \href{http://dx.doi.org/10.1016/j.aim.2007.11.012}{{\em Adv.Math.} {\bfseries
  217} (2008) 2664--2699},
\href{http://arxiv.org/abs/0707.1857}{{\ttfamily arXiv:0707.1857 [math.QA]}}.

\bibitem{Bushlanov:2009cv}
P.~Bushlanov, B.~Feigin, A.~Gainutdinov, and I.~Tipunin, ``{Lusztig limit of
  quantum sl(2) at root of unity and fusion of (1,p) Virasoro logarithmic
  minimal models},''
  \href{http://dx.doi.org/10.1016/j.nuclphysb.2009.03.016}{{\em Nucl.Phys.}
  {\bfseries B818} (2009) 179--195},
\href{http://arxiv.org/abs/0901.1602}{{\ttfamily arXiv:0901.1602 [hep-th]}}.

\bibitem{Ridout:2010qx}
D.~Ridout, ``{sl(2)$_{-1/2}$ and the Triplet Model},''
  \href{http://dx.doi.org/10.1016/j.nuclphysb.2010.03.018}{{\em Nucl. Phys.}
  {\bfseries B835} (2010) 314--342},
\href{http://arxiv.org/abs/1001.3960}{{\ttfamily arXiv:1001.3960 [hep-th]}}.

\bibitem{Abe:2007kb}
T.~Abe, ``{A Z(2)-orbifold model of the symplectic fermionic vertex operator
  superalgebra},'' {\em Math. Z.} {\bfseries 255} (2007) 755--792,
\href{http://arxiv.org/abs/math/0503472}{{\ttfamily arXiv:math/0503472
  [math.QA]}}.

\bibitem{Nagatomo:2009xp}
K.~Nagatomo and A.~Tsuchiya, ``{The Triplet Vertex Operator Algebra W(p) and
  the Restricted Quantum Group at Root of Unity},'' {\em Adv. Stdu. in Pure
  Math., Exploring new Structures and Natural Constructions in Mathematical
  Physics, Amer. Math. Soc.} {\bfseries 61} (2011) 1--49,
\href{http://arxiv.org/abs/0902.4607}{{\ttfamily arXiv:0902.4607 [math.QA]}}.

\bibitem{Feigin:2006xa}
B.~L. Feigin, A.~M. Gainutdinov, A.~M. Semikhatov, and I.~Y. Tipunin,
  ``{Kazhdan-Lusztig-dual quantum group for logarithmic extensions of Virasoro
  minimal models},'' \href{http://dx.doi.org/10.1063/1.2423226}{{\em J. Math.
  Phys.} {\bfseries 48} (2007) 032303},
\href{http://arxiv.org/abs/math/0606506}{{\ttfamily arXiv:math/0606506}}.

\bibitem{Rasmussen:2008xi}
J.~Rasmussen, ``{W-Extended Logarithmic Minimal Models},''
  \href{http://dx.doi.org/10.1016/j.nuclphysb.2008.07.029}{{\em Nucl. Phys.}
  {\bfseries B807} (2009) 495--533},
\href{http://arxiv.org/abs/0805.2991}{{\ttfamily arXiv:0805.2991 [hep-th]}}.

\bibitem{Eberle:2006zn}
H.~Eberle and M.~Flohr, ``{Virasoro representations and fusion for general
  augmented minimal models},''
  \href{http://dx.doi.org/10.1088/0305-4470/39/49/012}{{\em J. Phys. A}
  {\bfseries 39} (2006) 15245},
\href{http://arxiv.org/abs/hep-th/0604097}{{\ttfamily arXiv:hep-th/0604097}}.

\bibitem{Gaberdiel:2010rg}
M.~R. Gaberdiel, I.~Runkel, and S.~Wood, ``{A Modular invariant bulk theory for
  the c=0 triplet model},''
  \href{http://dx.doi.org/10.1088/1751-8113/44/1/015204}{{\em J.Phys. A}
  {\bfseries 44} (2011) 015204},
\href{http://arxiv.org/abs/1008.0082}{{\ttfamily arXiv:1008.0082 [hep-th]}}.

\bibitem{Pearce:2010pa}
P.~A. Pearce and J.~Rasmussen, ``{Coset Graphs in Bulk and Boundary Logarithmic
  Minimal Models},''
  \href{http://dx.doi.org/10.1016/j.nuclphysb.2011.01.014}{{\em Nucl.Phys.}
  {\bfseries B846} (2011) 616--649},
\href{http://arxiv.org/abs/1010.5328}{{\ttfamily arXiv:1010.5328 [hep-th]}}.

\bibitem{Adamovic:2009xs}
D.~Adamovic and A.~Milas, ``{On W-algebras associated to (2,p) minimal models
  and their representations},''
  \href{http://dx.doi.org/10.1093/imrn/rnq016}{{\em Int. Math. Res. Notices}
  {\bfseries 2010} (2009) },
\href{http://arxiv.org/abs/0908.4053}{{\ttfamily arXiv:0908.4053 [math.QA]}}.

\bibitem{Adamovic:2010zk}
D.~Adamovic and A.~Milas, ``{The Structure of Zhu's algebras for certain
  W-algebras},'' \href{http://dx.doi.org/10.1016/j.aim.2011.05.007}{{\em Adv.
  Math.} {\bfseries 227} (2010) 2425–2456},
\href{http://arxiv.org/abs/1006.5134}{{\ttfamily arXiv:1006.5134 [math.QA]}}.

\bibitem{Adamovic:2011xq}
D.~Adamovic and A.~Milas, ``{On W-algebra extensions of (2,p) minimal models: p
  \textgreater 3},''
  \href{http://dx.doi.org/10.1016/j.jalgebra.2011.07.006}{{\em J. Alg.}
  {\bfseries 344} (2011) 313--332},
\href{http://arxiv.org/abs/1101.0803}{{\ttfamily arXiv:1101.0803 [math.QA]}}.

\bibitem{Huang:2010pm}
Y.-Z. Huang, J.~Lepowsky, and L.~Zhang, ``{Logarithmic Tensor Category Theory
  for Generalized Modules for a Conformal Vertex Algebra, I: Introduction and
  Strongly Graded Algebras and their Generalized Modules},''
\href{http://arxiv.org/abs/1012.4193}{{\ttfamily arXiv:1012.4193 [math.QA]}}.

\bibitem{Huang:2010pn}
Y.-Z. Huang, J.~Lepowsky, and L.~Zhang, ``{Logarithmic tensor category theory,
  II: Logarithmic formal calculus and properties of logarithmic intertwining
  operators},''
\href{http://arxiv.org/abs/1012.4196}{{\ttfamily arXiv:1012.4196 [math.QA]}}.

\bibitem{Huang:2010pp}
Y.-Z. Huang, J.~Lepowsky, and L.~Zhang, ``{Logarithmic tensor category theory,
  III: Intertwining maps and tensor product bifunctors},''
\href{http://arxiv.org/abs/1012.4197}{{\ttfamily arXiv:1012.4197 [math.QA]}}.

\bibitem{Huang:2010pq}
Y.-Z. Huang, J.~Lepowsky, and L.~Zhang, ``{Logarithmic tensor category theory,
  IV: Constructions of tensor product bifunctors and the compatibility
  conditions},''
\href{http://arxiv.org/abs/1012.4198}{{\ttfamily arXiv:1012.4198 [math.QA]}}.

\bibitem{Huang:2010pr}
Y.-Z. Huang, J.~Lepowsky, and L.~Zhang, ``{Logarithmic tensor category theory,
  V: Convergence condition for intertwining maps and the corresponding
  compatibility condition},''
\href{http://arxiv.org/abs/1012.4199}{{\ttfamily arXiv:1012.4199 [math.QA]}}.

\bibitem{Huang:2010ps}
Y.-Z. Huang, J.~Lepowsky, and L.~Zhang, ``{Logarithmic tensor category theory,
  VI: Expansion condition, associativity of logarithmic intertwining operators,
  and the associativity isomorphisms},''
\href{http://arxiv.org/abs/1012.4202}{{\ttfamily arXiv:1012.4202 [math.QA]}}.

\bibitem{Huang:2011fn}
Y.-Z. Huang, J.~Lepowsky, and L.~Zhang, ``{Logarithmic tensor category theory,
  VII: Convergence and extension properties and applications to expansion for
  intertwining maps},''
\href{http://arxiv.org/abs/1110.1929}{{\ttfamily arXiv:1110.1929 [math.QA]}}.

\bibitem{Huang:2011fp}
Y.-Z. Huang, J.~Lepowsky, and L.~Zhang, ``{Logarithmic tensor category theory,
  VIII: Braided tensor category structure on categories of generalized modules
  for a conformal vertex algebra},''
\href{http://arxiv.org/abs/1110.1931}{{\ttfamily arXiv:1110.1931 [math.QA]}}.

\bibitem{Huang:2009xq}
Y.-Z. Huang, ``{Representations of vertex operator algebras and braided finite
  tensor categories},'' {\em Vertex Operator Algebras and Related Topics, An
  International Conference in Honor of Geoffrey Mason's 60th Birthday, ed. M.
  Bergvelt, G. Yamskulna and W. Zhao, Contemporary Math., Amer. Math. Soc}
  {\bfseries 497} (2009) 97--111,
\href{http://arxiv.org/abs/0903.4233}{{\ttfamily arXiv:0903.4233 [math.QA]}}.

\bibitem{Huang:2009mj}
Y.-Z. Huang, ``{Cofiniteness conditions, projective covers and the logarithmic
  tensor product theory},'' {\em J. Pure Appl. Alg.} {\bfseries 213} (2009)
  458--475,
\href{http://arxiv.org/abs/0712.4109}{{\ttfamily arXiv:0712.4109 [math.QA]}}.

\bibitem{Kondo:2009}
H.~Kondo and Y.~Saito, ``{Indecomposable decomposition of tensor products of
  modules over the restricted quantum universal enveloping algebra associated
  to sl$_2$},'' {\em J. Alg.} {\bfseries 330} (2011) 103--129,
  \href{http://arxiv.org/abs/math.QA/0901.4221}{{\ttfamily
  arXiv:math.QA/0901.4221}}.

\bibitem{Nagatomo:2002}
K.~Nagatomo and A.~Tsuchiya, ``{Conformal field theories associated to regular
  chiral vertex operator algebras, I: Theories over the projective line},''
  \href{http://dx.doi.org/10.1215/S0012-7094-04-12831-3}{{\em Duke Math J.}
  {\bfseries 128} (2005) 393--471},
  \href{http://arxiv.org/abs/math.QA/0206223}{{\ttfamily
  arXiv:math.QA/0206223}}.

\bibitem{Matsuo:2005}
A.~Matsuo, K.~Nagatomo, and A.~Tsuchiya, ``{Quasi-finite algebras graded by
  Hamiltonian and vertex operator algebras},'' {\em London Mathematical
  Society, Lecture Note Series} {\bfseries 372} (2010) 282 -- 329,
  \href{http://arxiv.org/abs/math.QA/0505071}{{\ttfamily
  arXiv:math.QA/0505071}}.

\bibitem{Miamoto:2004}
M.~Miyamoto, ``{Modular invariance of vertex operator algebras satisfying
  C(2)-cofiniteness},'' {\em Duke Math. J.} {\bfseries 122} (2004) 51--91.

\bibitem{Kazhdan:1994}
D.~Kazhdan and G.~Lusztig, ``{Tensor structures arising from affine Lie
  algebras IV},''
  \href{http://dx.doi.org/10.1090/S0894-0347-1994-1239507-1}{{\em J. Amer.
  Math. Soc.} {\bfseries 7} (1994) 383--453}.

\bibitem{Joyal:1993}
A.~Joyal and R.~Street, ``{Braided tensor categories},''
  \href{http://dx.doi.org/doi:10.1006/aima.1993.1055}{{\em Adv. Math.}
  {\bfseries 102} (1993) 20--78}.

\bibitem{Tsuchiya:1987}
A.~Tsuchiya and Y.~Kanie, ``{Vertex operators in the conformal field theory on
  P1 and monodromy representations of the braid group},''
  \href{http://dx.doi.org/10.1007/BF00401159}{{\em Lett. Math. Phys.}
  {\bfseries 13} (1987) 303--312}.

\bibitem{DiFrancesco:1997nk}
P.~Di~Francesco, P.~Mathieu, and D.~Senechal, ``{Conformal field theory},''
  {\em Springer} (1997) 890.

\bibitem{Huang:2008}
Y.-Z. Huang, ``{Rigidity and modularity of vertex tensor categories},'' {\em
  Comm. Contemp. Math.} {\bfseries 10} (2008) 871--911,
  \href{http://arxiv.org/abs/0502533}{{\ttfamily arXiv:0502533 [math.QA]}}.

\bibitem{Nahm:1994by}
W.~Nahm, ``{Quasirational fusion products},''
  \href{http://dx.doi.org/10.1142/S0217979294001597}{{\em Int. J. Mod. Phys.}
  {\bfseries B8} (1994) 3693--3702},
\href{http://arxiv.org/abs/hep-th/9402039}{{\ttfamily arXiv:hep-th/9402039}}.

\bibitem{Gaberdiel:1993mt}
M.~Gaberdiel, ``{Fusion rules of chiral algebras},''
  \href{http://dx.doi.org/10.1016/0550-3213(94)90540-1}{{\em Nucl. Phys.}
  {\bfseries B417} (1994) 130--150},
\href{http://arxiv.org/abs/hep-th/9309105}{{\ttfamily arXiv:hep-th/9309105}}.

\bibitem{Hashimoto:2012}
Y.~Hashimoto and A.~Tsuchiya. In preparation.

\bibitem{Fel:1989}
G.~Felder, ``{BRST approach to minimal models},'' {\em Nuc. Phy. B} {\bfseries
  317} (1989) 215--236.

\bibitem{Tsuchiya:1986}
A.~Tsuchiya and Y.~Kanie, ``{Fock space representations of the Virasoro algebra
  -- Intertwining operators},'' {\em Publ. RIMS, Kyoto Univ.} {\bfseries 22}
  (1986) 259--327.

\bibitem{Abramowitz:1964}
M.~Abramowitz and I.~Stegun, ``{Handbook of Mathematical Functions},'' {\em
  Dover Publications} (1964) .

\end{thebibliography}\endgroup

\end{document}